\documentclass[aps,12pt,tightenlines,showkeys]{revtex4}

\usepackage{graphicx}
\usepackage{amsmath}
\usepackage{bbm}
\usepackage{amsthm}
\usepackage{pdflscape}

\def\<{\left\langle}
\def\>{\right\rangle}
\def\tab{\hskip .166667in}
\def\tabb{\hskip .333333in}

\newtheorem{theorem}{Theorem}
\newtheorem{lemma}{Lemma}
\newtheorem{definition}{Definition}
\theoremstyle{remark}
\newtheorem{remark}{Remark}

\begin{document}

\title{Dense Crystalline Dimer Packings of Regular Tetrahedra}

\author{Elizabeth R. Chen}
\email{bethchen@umich.edu}
\affiliation{Department of Mathematics, University of Michigan, Ann Arbor
  Michigan 48109, USA}

\author{Michael Engel}
\email{engelmm@umich.edu}
\affiliation{Department of Chemical Engineering, University of Michigan,
  Ann Arbor Michigan 48109, USA}

\author{Sharon C. Glotzer}
\email{sglotzer@umich.edu}
\altaffiliation{Department of Materials Science and Engineering, University
  of Michigan, Ann Arbor Michigan 48109, USA}
\affiliation{Department of Chemical Engineering, University of Michigan,
  Ann Arbor Michigan 48109, USA}

\begin{abstract}
We present the densest known packing of regular tetrahedra with density $\phi
= \tfrac{4000}{4671} = 0.856347\ldots$ Like the recently discovered packings
of Kallus et al.\ and Torquato-Jiao, our packing is crystalline with a unit
cell of four tetrahedra forming two triangular dipyramids (dimer clusters). We
show that our packing has maximal density within a three-parameter family of
dimer packings. Numerical compressions starting from random configurations
suggest that the packing may be optimal at least for small cells with up to 16
tetrahedra and periodic boundaries.
\end{abstract}

\date{Submitted: December 27, 2009; Revised: \today}


\maketitle

\section{Introduction}

The problem of the packing of tetrahedra, which has modern-day applications to
such wide-ranging topics as metamaterials with novel optical properties,
nanomaterials \cite{kotov1,kotov2,kotov3}, and virus formation, hails back to
the early Greeks \cite{chenthesis}. Aristotle, in discussing the assignment of
geometrical figures to ``heavenly bodies'', mistakenly believed that regular
tetrahedra (``pyramids'') tile (Euclidean) space perfectly. Nearly 1800 years
later, Johannes M\"uller (aka Regiomontanus, 1436-1476) contradicted
Aristotle's claim \cite{struik25,senechal81}.  400 years after M\"uller,
Minkowski showed that the densest lattice packing of any convex body must
satisfy certain constraints \cite{minkowski96} and mistakenly argued that the
densest lattice packing of tetrahedra (i.e. one tetrahedron per unit cell) had
density $\phi = \tfrac{9}{38}$ \cite{minkowski12}. In 1900, Hilbert posed the
problem of the densest lattice packing of regular tetrahedra as a special case
of the 18th of his famous list of problems \cite{hilbert00}. Gr\"omer
conjectured \cite{gromer62}, and Hoylman later proved \cite{hoylman70}, that
the densest lattice packing of a single tetrahedron is $\phi =
\tfrac{18}{49}$.

In 1972, Stanislaw Ulam conjectured that spheres would have the lowest maximum
packing density of all convex bodies, including tetrahedra
\cite{gardner01}. That density is $\phi_{\rm
  sphere}=\tfrac{\pi}{\sqrt{18}}=0.740480\ldots$ \cite{hales05}. In 2000,
Betke and Henk developed an efficient computer algorithm to compute the
densest lattice packing of any convex body, and applied it to the Archimedean
solids \cite{betke00}. Conway and Torquato \cite{conway06} used Betke and
Henk's algorithm to compute the packing density of tetrahedra derived from the
densest lattice packing of icosahedra, and also examined other promising
packings.  All of them packed worse than spheres, with a maximum packing
fraction $\phi\approx 0.7175$. This led them to suggest the tantalizing
possibility that the Ulam conjecture might, in fact, fail for tetrahedra. They
further proposed that tetrahedra might have the lowest packing density of all
convex bodies. Motivated by this, Chaikin and coworkers \cite{chaikin07}
performed experiments on tetrahedral dice packed in spherical and cylindrical
containers.  They reported random packings with densities as high as
$0.75\pm0.03$ [16] and $0.76\pm0.02$ [17], with the error arising from the
rounded dice vertices and edges \cite{jaoshvili10}.

In 2008, Chen proposed a packing of nonamers (nine tetrahedra forming two
pentagonal dipyramids that share one tetrahedron) arranged in layers, whose
density ($\phi \approx 0.7786$) clearly exceeded, for the first time, the
maximum packing density of spheres \cite{chen08}. Since Chen's publication, a
flurry of activity has resulted in successively higher and higher packing
densities for tetrahedra. It was soon demonstrated that the nonamer crystal
could be numerically compressed to slightly higher packing densities ($\phi =
0.7837$ \cite{akbariAPS09} and $\phi = 0.7820$ \cite{torquatoNature09}).  The
first {\it disordered} (i.e.\ non-periodic) packing of tetrahedra to exceed
the maximum packing density of spheres was reported by Haji-Akbari et
al.\ \cite{akbariAPS09,akbariNature09}; they obtained a packing density of
$\phi = 0.7858$ using Monte Carlo simulation of systems of 8000 tetrahedra
compressed from random initial conditions.  Their disordered packing contains
a preponderance of pentagonal dipyramids, arranged randomly and thus
differently than in the Chen structure, along with other motifs such as
icosahedra.

Previously, all dense ordered packings proposed for tetrahedra were based on
analytical construction (or numerical compressions thereof).  Haji-Akbari et
al.\ showed, using MC simulations of initially random systems containing up to
nearly 22,000 tetrahedra, that at packing densities $\phi > 0.5$ an
equilibrium fluid of hard tetrahedra spontaneously transforms to a dodecagonal
quasicrystal, which can be compressed to $\phi = 0.8324$
\cite{akbariNature09}. By numerically constructing and then compressing four
unit cells of a periodic quasicrystal approximant with an 82-tetrahedron unit
cell, they obtained a packing density as high as $\phi = 0.8503$.

Motivated by a numerical search, Kallus, Elser and Gravel \cite{kallus09}
found a one-parameter family of dimer cluster packings obtained via an
analytical construction with a density of
$\phi=\tfrac{100}{117}=0.854700\ldots$, exceeding the density of the
quasicrystal approximant. Each dimer cluster contains two face-sharing
tetrahedra, and is equivalent to a triangular dipyramid (or bipyramid). Two
dimers comprise a single unit cell. Torquato and Jiao generalized the
analytical construction of Kallus et al.\ to a {\it two}-parameter family of
packings, and obtained an even denser packing with
$\phi=\tfrac{12250}{14319}=0.855506\ldots$ \cite{torquato09}.

In this paper, we further generalize the Kallus et al.\ and Torquato-Jiao
family of dimer cluster packings to consider a {\it three}-parameter family of
packings.  We obtain a maximum packing density of
$\phi=\tfrac{4000}{4671}=0.856347\ldots$ as the optimal solution within the
three-parameter family and show that the Kallus et al.\ and Torquato-Jiao
packings are special cases of our more general construction.  Furthermore, we
present numerical compression simulations of small systems of tetrahedra using
standard isobaric Monte Carlo with variable simulation cell shape.  For
systems with 4, 8, 12 and 16 tetrahedra with initially random configurations,
the analytically predicted dimer crystal with $\phi = 0.856347$ is recovered.

\section{Analytical construction}

\subsection{Double dimer configurations}

\begin{definition}
A \emph{dimer (cluster)} is a basic building block formed by two face-sharing
regular tetrahedra arranged into a triangular dipyramid. The vertices of a
\emph{positive dimer} are chosen up to translation as
\begin{equation}\begin{split}
o&=\<+2,+2,+2\>,\\
p=\<+2,-1,-1\>,\quad q&=\<-1,+2,-1\>,\quad r=\<-1,-1,+2\>,\\
s&=\<-2,-2,-2\>,
\end{split}\end{equation}
such that $p$, $q$, $r$ span the common face. A \emph{negative dimer} is
related to a positive dimer by inversion. The positive dimer with vertices
$o$, $p$, $q$, $r$, $s$ centered at the origin is called $+{\bf F}_2$; the
corresponding negative dimer with vertices $-o$, $-p$, $-q$, $-r$, $-s$
centered at the origin is called $-{\bf F}_2$.
\end{definition}

\begin{definition}
A crystalline arrangement of tetrahedra with four tetrahedra in the unit cell
forming one negative dimer and one positive dimer is called a \emph{double
  dimer configuration}. If the configuration is free of overlaps, then it is a
\emph{double dimer packing}.
\end{definition}

The set of possible double dimer configurations $\mathcal{P}$ is a set of
translates of the two dimers. We place the positive dimer $+{\bf F}_{2}$ at
the origin and translate it by the lattice ${\bf L}^{+}$. The negative dimer
$-{\bf F}_{2}$ is placed by the lattice coset ${\bf L}^{-}$.  Without loss of
generality we can assume that ${\bf L}^{+}$ is spanned by $a+b$, $b+c$, $c+a$
and choose the offset of ${\bf L}^{-}$ as $d+a$ (equivalently $d+b$, or
$d+c$), where the lattice vectors ($a$, $b$, $c$) and the offset vector
$d$ are free to vary.  This means that
\begin{equation}\begin{split}
{\bf L}^{+}&=\{n_aa+n_bb+n_cc\mid n_a+n_b+n_c=0\bmod{2}\},\\
{\bf L}^{-}&=\{n_aa+n_bb+n_cc\mid n_a+n_b+n_c=1\bmod{2}\}+d.
\end{split}\end{equation}
We say a dimer is part of the \emph{$n$-th layer} with $n=n_a+n_b+n_c$. The
even (odd) layers contain positive (negative) dimers.

The lattice and offset vectors contain 12 coordinate variables, which means
the set $\mathcal{P}$ is a 12-dimensional linear space. However many of the
configurations have overlaps, and thus are not packings.

\subsection{Three-parameter family}

In order to verify that the dimers of a double dimer configuration are not
overlapping (i.e.\ that it is a packing), intersections of dimers with their
neighbors are studied. It is sufficient to consider the dimer at the origin:
\begin{lemma}
The space group of a double dimer configuration acts transitively on the
dimers.
\end{lemma}
\begin{proof}
The inversion $x\mapsto d+a-x$ maps positive dimers to negative dimers.
\end{proof}

Intersections can occur between three types of simplices forming the surface
of a dimer: triangular faces, edges, and vertices. Given two simplices, we say
the \emph{incidence condition} for these simplices holds if their intersection
is non-empty. Incidence conditions are of either face-to-face, face-to-edge,
face-to-vertex, edge-to-edge, edge-to-vertex or vertex-to-vertex type.

Next, we replace the simplices by affine subspaces and use the notation
$V[p]=\{p\}$ to be the vertex $p$,
$E[p,q]=\{p+\kappa(q-p):\kappa\in\mathbbm{R}\}$ to be the line containing the
edge from $p$ to $q$, and
$F[p,q,r]=\{p+\kappa(q-p)+\lambda(r-p):\kappa,\lambda\in\mathbbm{R}\}$ to be
the plane containing the triangular face spanned by $p$, $q$, and $r$. We say
the \emph{linear incidence condition} holds if the intersection of two affine
subspaces is non-empty. In the following we will consider linear incidence
conditions only. They are easier to handle analytically and can, in certain
cases, be used as sufficient conditions for the absence of overlaps.

\begin{lemma}\label{lem:reducedim}
Given the nine linear incidence conditions
\begin{equation}\begin{split}
{\bf G}^{+}_a:&\quad F[o,p,q]\cap(d-F[o,p,q]+a)\neq\emptyset,\\
{\bf G}^{+}_b:&\quad F[o,q,r]\cap(d-F[o,q,r]+b)\neq\emptyset,\\
{\bf G}^{+}_c:&\quad F[o,r,p]\cap(d-F[o,r,p]+c)\neq\emptyset,\\
{\bf G}^{+}_{a+b+c}:&\quad F[o,r,p]\cap(d-F[o,r,p]+a+b+c)\neq\emptyset,\\
{\bf G}^{-}_a:&\quad F[s,r,q]\cap(d-F[s,r,q]-a)\neq\emptyset,\\
{\bf G}^{-}_b:&\quad F[s,p,r]\cap(d-F[s,p,r]-b)\neq\emptyset,\\
{\bf G}^{-}_c:&\quad F[s,q,p]\cap(d-F[s,q,p]-c)\neq\emptyset,\\
{\bf G}^{-}_{a+b+c}:&\quad F[s,q,p]\cap(d-F[s,q,p]-a-b-c)\neq\emptyset,\\
{\bf G}_{a+b}:&\quad E[o,q]\cap(E[s,r]+a+b)\neq\emptyset,
\end{split}\end{equation}
the collection $\mathcal{P}' = \{p\in \mathcal{P}: \text{the conditions } {\bf
  G}^{\pm}_a, {\bf G}^{\pm}_b, {\bf G}^{\pm}_c, {\bf G}^{\pm}_{a+b+c}, {\bf
  G}_{a+b} \text{ hold}\}$ is a three-dimensional linear space.
\end{lemma}

\begin{proof}
The nine linear incidence conditions are sufficient for the statement that the
dimer at the origin is not overlapping with the neighbors at $d\pm a$, $d\pm
b$, $d\pm c$, $d\pm(a+b+c)$, and $\pm(a+b)$. Note that ${\bf L}^+$ is
invariant under inversion at the origin, so that we do not have to consider
the neighbors at $+(a+b)$ and $-(a+b)$ separately.

Since the planes involved in the face-to-face condition ${\bf G}^+_a$ are
parallel, they coincide if they are incident, so we can alternatively write
${\bf G}^+_a$ as an equation:
\begin{equation}
{\bf G}^+_a:\quad 2o + \kappa(p-o)+\lambda(q-o)=d+a
\end{equation}
with parameters $\kappa,\lambda\in\mathbbm{R}$. This equation is a vector
equation and thus there is one scalar equation for each of its three
components. Two of the scalar equations are used to eliminate $\kappa$ and
$\lambda$. The remaining scalar equation is a single linear constraint for the
components of the vectors $d$ and $a$. Similar constraints are obtained for
the seven other face-to-face conditions. The edge-to-edge condition ${\bf
  G}_{a+b}$ corresponds to the equation
\begin{equation}
{\bf G}_{a+b}:\quad o+\mu(q-o)=s+\nu(r-s)+a+b
\end{equation}
with parameters $\mu,\nu\in\mathbbm{R}$. Again, a single linear constraint for
the components of the vectors $a$ and $b$ is obtained. The system of linear
constraints turns out to have rank 9, so the dimension of the space of double
dimer configurations $\mathcal{P}$ is reduced from 12 to 3. With the new
scalar parameters $u$, $v$, $w$, the lattice vectors and the offset vector can
be written as
\begin{equation}\label{Eq:generalbasis}\begin{split}
a &= \<\tfrac{27}{10}+u,\tfrac{21}{20}-v,-\tfrac{3}{20}+2u+v\>,\\
b &= \<-\tfrac{3}{10}-u,\tfrac{51}{20}+v,\tfrac{27}{20}-2u-v\>,\\
c &= \<\tfrac{129}{160}-u+2v+2w,-\tfrac{237}{320}+\tfrac{1}{2}u-v+3w,\tfrac{753}{320}+\tfrac{1}{2}u-v+w\>,\\
d &= \<\tfrac{1}{10}+u,-\tfrac{1}{20}+u+v,-\tfrac{1}{20}+u-v\>.
\end{split}\end{equation}
\end{proof}

\begin{remark}
(i) The eight face-to-face incidence conditions are important for dense
  packings, because the corresponding faces are parallel and therefore can be
  matched perfectly.

(ii) The face-to-face constraints imposed by ${\bf G}^+_{a+b+c}$ and ${\bf
    G}^-_{a+b+c}$ involve the faces spanned by $o$, $r$, $p$ and $s$, $q$,
  $p$, respectively. We balance the constraints with an opposing edge-to-edge
  constraint with the neighbors at $\pm(a+b)$ on the `opposite' edge spanned
  by $o$, $q$ and $s$, $r$ in the form of the incidence condition ${\bf
    G}_{a+b}$.
\end{remark}

\subsection{Restricted parameter space}

We now specify a subset $\mathcal{P}''\subset\mathcal{P}'$ of the
three-parameter space that consists of double dimer packings only, and we
later optimize density over this subset.

\begin{lemma}
Consider the restricted parameter space
\begin{equation}
\mathcal{P}'' = \left\{\<u,v,w\>\in \mathcal{P}':\ 
\left|\tfrac{1}{2}u+2v\right|\leq\tfrac{33}{320}+w\ \wedge\
\left|v\right|\leq\tfrac{3}{64}-w\right\},
\end{equation}
where $\<u,v,w\>$ is the parameterization in (\ref{Eq:generalbasis}). Each
double dimer configuration $p\in\mathcal{P}''$ is a packing.
\end{lemma}

\begin{proof}
Since double dimer packings are transitive, it is sufficient to look for
overlaps of $+{\bf F}_2$ only. We write a positive dimer as $+{\bf
  F}_2+n_aa+n_bb+n_cc$ and a negative dimer as $d-{\bf
  F}_2+n_aa+n_bb+n_cc$. The first part of the proof shows that we only need to
consider the 26 neighbors with $-1\leq n_a,n_b,n_c\leq 1$ when looking for
overlaps. Note that $\mathcal{P}''$ has the shape of a sheared tetrahedron
with the four vertices $e_{1,2}=\<\pm3/10,0,+3/64\>$ and
$e_{3,4}=\<\mp3/5,\pm3/20,-33/320\>$. We verify that for each configuration
corresponding to one of the four extremal points $e_j$, the plane
$a+\mathbbm{R}b+\mathbbm{R}c$ separates $+{\bf F}_2$ from $+{\bf F}_2+2a$ and
the plane $d/2+a+\mathbbm{R}b+\mathbbm{R}c$ separates $+{\bf F}_2$ from
$d-{\bf F}_2+2a$. This shows that $n_a\leq 1$. By similar arguments we get
$n_a\geq-1$ and $-1\leq n_b,n_c\leq 1$. Since $a$, $b$, $c$ are linear in $u$,
$v$, $w$ and $\mathcal{P}''$ is convex, the separating planes and the
restrictions $-1\leq n_a,n_b,n_c\leq 1$ remain valid for every double dimer
configuration in $\mathcal{P}''$.

By construction of $\mathcal{P}'$ in Lemma \ref{lem:reducedim} we know that
the ten dimers $d-{\bf F}_2\pm a$, $d-{\bf F}_2\pm b$, $d-{\bf F}_2\pm c$,
$d-{\bf F}_2\pm(a+b+c)$, $+{\bf F}_2\pm(a+b)$ do not overlap with $+{\bf
  F}_2$. Next, we find that the dimers $+{\bf F}_2$ and $d-{\bf F}_2+a+b-c$
are separated by the plane $F[o,p,q]$ for any configuration in
$\mathcal{P}''$. We find similar planes for $d-{\bf F}_2+a-b+c$, $d-{\bf
  F}_2-a+b+c$, $d-{\bf F}_2-a-b+c$, $d-{\bf F}_2-a+b-c$, and $d-{\bf
  F}_2+a-b-c$, which shows that these six dimers cannot overlap with $+{\bf
  F}_2$. Finally note that pairs of positive neighbors are related by
inversion symmetry. These considerations reduce the number of neighbors that
have to be checked from 26 to five. The remaining ones are: $+{\bf F}_2+b+c$,
$+{\bf F}_2+c+a$, $+{\bf F}_2+a-b$, $+{\bf F}_2+b-c$, $+{\bf F}_2+c-a$.

Consider the linear edge-to-edge and vertex-to-vertex incidence conditions
\begin{equation}\begin{split}
{\bf H}_{b+c}:&\quad E[o,r]\cap(E[s,p]+b+c)\neq\emptyset,\\
{\bf H}_{c+a}:&\quad E[o,p]\cap(E[s,q]+c+a)\neq\emptyset,\\
{\bf H}_{a-b}:&\quad V[p]\cap(E[r,q]+a-b)\neq\emptyset,\\
{\bf H}_{b-c}:&\quad V[q]\cap(F[s,p,r]+b-c)\neq\emptyset,\\
{\bf H}_{c-a}:&\quad V[r]\cap(F[o,q,p]+c-a)\neq\emptyset.\\
\end{split}\end{equation}
Similar to the linear incidence conditions in Lemma \ref{lem:reducedim}, the
conditions ${\bf H}_{b+c}$, ${\bf H}_{c+a}$, ${\bf H}_{b-c}$, ${\bf H}_{c-a}$
can be written as four linear constraints for the parameters $u$, $v$, $w$,
which correspond to four \emph{boundary planes}. Each of the planes separates
$\mathcal{P}'$ into two half-spaces. Overlaps can only appear in one of the
half-spaces, and the half-spaces that are free from overlaps are:
\begin{equation}\begin{split}
{\bf H}_{b+c}:&\quad +\tfrac{1}{2}u+2v-w\leq\tfrac{33}{320},\\
{\bf H}_{c+a}:&\quad -\tfrac{1}{2}u-2v-w\leq\tfrac{33}{320},\\
{\bf H}_{b-c}:&\quad -v+w\leq\tfrac{3}{64},\\
{\bf H}_{c-a}:&\quad +v+w\leq\tfrac{3}{64}.
\end{split}\end{equation}
The intersection of the four half-spaces is the restricted parameter space
$\mathcal{P}''$.

The condition ${\bf H}_{a-b}$ is different from the other ones, since it is a
vertex-to-edge condition. It is true if and only if both vertex-to-face
conditions $V[p]\cap(F[o,r,q]+a-b)\neq\emptyset$ and
$V[p]\cap(F[s,r,q]+a-b)\neq\emptyset$ are true. We calculate the region of
parameter space that is free from overlap using the half-spaces for these two
conditions:
\begin{equation}
{\bf H}_{a-b}:\quad -u\leq0\quad\wedge\quad +u\leq0,
\end{equation}
which covers all of $\mathcal{P}'$. Hence, the vertex $V[p]$ can be incident
to the edge $E[r,q]$ only along the \emph{central plane} $u=0$, and ${\bf
  H}_{a-b}$ is never needed to prevent overlaps within $\mathcal{P}'$.
\end{proof}

\begin{lemma}
The collection $\mathcal{P}''$ contains exactly two maximal density packings
with density $\tfrac{4000}{4671}$. The packings are related by a
crystallographic symmetry operation.
\end{lemma}

\begin{proof}
The unit cell volume as a function of the parameters $\<u,v,w\>$,
\begin{equation}
V=\left|\det[a+b,b+c,c+a]\right|=\tfrac{9}{25}(117+60u^{2}-80uv-80v^{2}),
\end{equation}
is a hyperbolic paraboloid with saddle point at $u=v=0$. The parameters
$\<u,v,w\>$ can be chosen such that $V$ does not depend on $w$, which means
packings corresponding to the line $\<0,0,w\>$ are related by a lattice
shear. Surfaces of equal volume are hyperbolic cylinders. The extrema of
$V|_{\mathcal{P}''}$ are taken on the boundary. We find that the densest
packings are located on the \emph{maximal lines}:
\begin{equation}\begin{split}
{\bf H}_{b+c}\wedge{\bf H}_{c-a}:&\quad
+\tfrac{3}{320}\<2,5,0\>+\tfrac{1}{6}u\<6,-1,+1\>,\\
{\bf H}_{b-c}\wedge{\bf H}_{c+a}:&\quad
-\tfrac{3}{320}\<2,5,0\>+\tfrac{1}{6}u\<6,-1,-1\>.
\end{split}\end{equation}
Maximizing the packing density $\phi=2U/V$ with the tetrahedron volume $U =
\tfrac{1}{6}\left|\det[o-p,o-q,o-r]\right| = 9$ along the maximal lines yields
$\phi = \tfrac{4000}{4671}$ and two \emph{optimal points}:
\begin{equation}
\<u_{\pm},v_{\pm},w_{\pm}\> = \pm\tfrac{3}{320}\<2,5,0\>.
\end{equation}

It can be shown that for two packings with parameters $\<u,v,w\>$ and $\<\bar
u,\bar v,\bar w\>$ related by $\<\bar u,\bar v,\bar w\> = \<-u,-v,w\>$ there
exists a direct isometry
\begin{equation}\label{Eq:symmetry}
T = \tfrac{1}{3}\left(\begin{matrix}
+1 & -2 & -2\\ -2 & -2 & +1\\ -2 & +1 & -2
\end{matrix}\right)
\end{equation}
that symmetrically maps the lattice vectors $Ta = -\bar b$, $Tb = -\bar a$,
$Tc = -\bar c$, $Td = \bar d$ and the vertices of the dimers $To = \bar s$,
$Tp = \bar p$, $Tq = \bar r$, $Tr = \bar q$, $Ts = \bar o$. In particular, $T$
maps the optimal points onto each other.
\end{proof}

\subsection{A dense double dimer packing}\label{Sec:dense}

The following theorem is the main result of our work. For the classification
of crystallographic point groups and space groups we follow the terminology in
\cite{hahn02}.
\begin{theorem}\label{thm:dense}
There exists a double dimer packing of tetrahedra with packing density
$\tfrac{4000}{4671}$. Its space group is $\mathrm{P}\bar{1}$ (point group
$\bar{1}$) and acts transitively on the dimers.
\end{theorem}
\begin{proof}
The double dimer packing at the optimal point $\<u_{+},v_{+},w_{+}\>$ has
density $\tfrac{4000}{4671}$. It is specified by the lattice vectors
\begin{equation}\label{Eq:basis}\begin{split}
a &= \tfrac{3}{320}\<290,107, -7\>,\\
b &= \tfrac{3}{320}\<-34,277,135\>,\\
c &= \tfrac{3}{320}\< 94,-83,247\>,
\end{split}\end{equation}
and the offset vector
\begin{equation}\label{Eq:basis2}
d = \tfrac{1}{320}\< 38,  5,-25\>.
\end{equation}
We know from the angles between the lattice vectors that the crystal system is
triclinic. The only non-translational symmetries are inversions mapping ${\bf
  L}^{+}$ onto ${\bf L}^{-}$, which determines the space group.
\end{proof}

\begin{remark}\label{rem:symmetry}
(i) The lattice spanned by the vectors $a$, $b$, and $c$ in
  Theorem~\ref{thm:dense} can be obtained from a simple cubic lattice by a
  small deformation. Furthermore, the norm of $d$ is much smaller than the
  norms of $a$, $b$, $c$. These two observations mean that the set ${\bf
    L}={\bf L}^{+}\cup{\bf L}^{-}$ is structurally similar to the rock-salt
  lattice.

(ii) One of the axes that is almost a three-fold symmetry is the diagonal
  $a+b+c$, which coincides roughly with the dimer axis.  By the choice of
  basis vectors in (\ref{Eq:basis}) and (\ref{Eq:basis2}), the three-fold
  symmetry of the individual dimers is broken.
\end{remark}

\subsection{Comparison with previously found double dimer packings}

The one-parameter family of dense tetrahedron packings found by Kallus, Elser
and Gravel \cite{kallus09} has constant packing density $\phi =
\tfrac{100}{117}$ and corresponds to the \emph{symmetric line} $\<0,0,w\>$. The
packings are transitive on individual tetrahedra as a result of the symmetry
$T$, a two-fold rotation around the axis through the vertex $p=\<2,-1,-1\>$.

Torquato-Jiao's two-parameter family of packings \cite{torquato09} corresponds
to the generic plane $5u=-2v$. The intersections of this plane with the
maximal lines gives the parameters
$\tfrac{3}{2240}(\pm\<16,-40,0\>-\<0,0,5\>)$ and the maximum packing density
$\phi = \tfrac{12250}{14319}$. The two packings are not related by symmetry.

\section{Numerical compression of small cells}\label{Sec:numerics}

\subsection{Motivation and methods}

The packing given in Theorem \ref{thm:dense} is optimal under two assumptions:
(1) the densest tetrahedron packing is a double dimer packing, and (2) the
packing is a configuration in the restricted parameter space
$\mathcal{P}''$. To investigate the possibility of denser tetrahedron packings
-- double dimer packings as well as different types of packings -- we rely on
a numerical search.

We use standard Monte Carlo simulation \cite{frenkelSmit}, which allows a
system of tetrahedra to find dense packings by stochastically exploring all
possible configurations subject to the laws of statistical mechanics. An
elementary simulation step consists of a random displacement move within a
finite simulation cell -- taking into account periodic boundary conditions --
and a random rotation move of randomly chosen tetrahedra. The move is rejected
if it generates an overlap, or accepted otherwise. The starting configuration
is a dilute, random arrangement. During the simulation run, the system is
slowly compressed by rescaling the size of the simulation cell. The
compression is controlled by applying external forces in the isobaric
ensemble. Additionally, fluctuations of cell shape by shearing the cell in
random directions by random amounts are allowed. A lattice reduction technique
minimizes the distortion of the simulation cell after each shear. For details
of the tetrahedron overlap detection algorithm we refer to method two in the
Methods section of Ref. \cite{akbariNature09}.

The Monte Carlo scheme samples the high-dimensional configuration space
stochastically, and is not biased towards any particular type of packing. The
only constraint is the number of tetrahedra in the simulation cell, $N$, which
does not change during the simulation. In the following, the search is
restricted to small cells, $1\leq N\leq 16$, where efficient compressions are
easily and rapidly possible to high accuracy.

\subsection{Monte Carlo simulation results}

To obtain sufficient statistics, $M=400$ compression simulations were run for
each value of $N$ (1000 runs for $N=16$). Each run involves 7 million Monte
Carlo moves per particle and results in a final density $\phi_i$,
$i=1,\ldots,M$ for a given $N$. The \emph{maximum numerical density} for a
given $N$ is $\hat{\phi}=\max\{\phi_i\}$. The distribution of $\phi_i$ close
to $\hat{\phi}$ indicates the ease with which we can obtain the optimal
packing in simulation. For most $N$ we find a clear gap separating a set of
very dense packings from the rest in the sense that the relative density
variation among the very dense packings is significantly smaller than the
gap between very dense and less dense packings. In this case, we denote the simulations corresponding to the very dense
packings as \emph{successful}.

The maximum numerical densities are given in Table \ref{Tab:numcomp} and their
corresponding packings are depicted in Figure \ref{Fig:numfigs}. As can be
seen from the success rates in the Table, the geometrically constructed
optimal packings for $N=1,2$ are obtained very easily and in every
simulation. The $N=3$ packing is three-fold symmetric and, as far as we know,
has not yet been reported in the literature. Its optimal density
$\phi=\tfrac{2}{3}$ can be calculated analytically. The structures with
$N=4,8,12,16$ are the dimer packings discussed in section \ref{Sec:dense}. The
packing with $N=5$ consists of imperfect pentamers, i.e.\ four tetrahedra
arranged face-to-face to a central one. $N=6$ is a mixture of dimers and
single tetrahedra (monomers), and $N=7$ is identical to $N=8$ with one
complete vacancy (missing tetrahedron).
\begin{table}
\centering
\setlength{\tabcolsep}{9pt}
\renewcommand{\arraystretch}{1.2}
\begin{tabular}{c|cc|c|l}
\#Tetra & \multicolumn{2}{c|}{Maximum Density} & Success & Motifs, \\
$N$     & Numerical, $\hat{\phi}$ & Analytical, $\phi$ & Rate & Structural Description\\
\hline
 1 & 0.367346 & $18/49$ & 100\% & 1 monomer \cite{hoylman70} \\
 2 & 0.719486 & $\phi_2$ & 100\% & 2 monomers, transitive \cite{kallus09} \\
 3 & 0.666665 & $2/3$ & 21\% & 3 monomers, three-fold symmetric \\
 4 & 0.856347 & $4000/4671$ & 80\% & 2 dimers (positive + negative)\\
 5 & 0.748096 & $\phi_5$ & 22\% & 1 pentamer, asymmetric \\
 6 & 0.764058 & $\phi_6$ & 11\% & 2 dimers + 2 monomers\\
 7 & 0.749304 & $3500/4671$ & 15\% & $2\times2$ dimers minus 1 monomer\\
 8 & 0.856347 & $4000/4671$ & 44\% & $2\times2$ dimers, identical to $N=4$\\
 9 & 0.766081 & & --- & 1 pentagonal dipyramid + 2 dimers\\
10 & 0.829282 & $\phi_{10}$ & 2\% & 2 pentagonal dipyramids\\
11 & 0.794604 & & --- & 1 nonamer + 2 monomers\\
12 & 0.856347 & $4000/4671$ & 3\% & $3\times2$ dimers, identical to $N=4$\\
13 & 0.788728 & & 4\% & 1 pentagonal dipyramid + 4 dimers\\
14 & 0.816834 & & 3\% & 2 pentagonal dipyramids + 2 dimers\\
15 & 0.788693 & & --- & Disordered, non-optimal\\
16 & 0.856342 & $4000/4671$ & $<1\%$ & $4\times2$ dimers, identical to $N=4$\\
$\vdots$ & $\vdots$ & & & \qquad\qquad\qquad$\vdots$\\
$8\times 82$ & 0.850267 & & & Quasicrystal approximant \cite{akbariNature09}\\
\end{tabular}
\caption{Maximum numerical densities $\hat{\phi}$ for packings with small
  cells, obtained with numerical compression via Monte Carlo compression
  starting from a random configuration. A data file with the packings may be
  downloaded from the internet at \cite{wwwtetra}. For comparison, the
  quasicrystal approximant result with $N=8\times82$ is included. Details
  about the analytical results $\phi_2=9/\left(139-40\sqrt{10}\right)$,
  $\phi_5=0.74809657\ldots$, $\phi_6=11228544/\left(97802181 -
  132043\sqrt{396129}\right)$, and
  $\phi_{10}=29611698560/\left(23657426736+4919428689\sqrt{6}\right)$ are
  given in Appendix D of Ref. \cite{addinfoarxiv}.}
\label{Tab:numcomp}
\end{table}
\begin{figure}
\centering
\includegraphics[width=0.95\columnwidth]{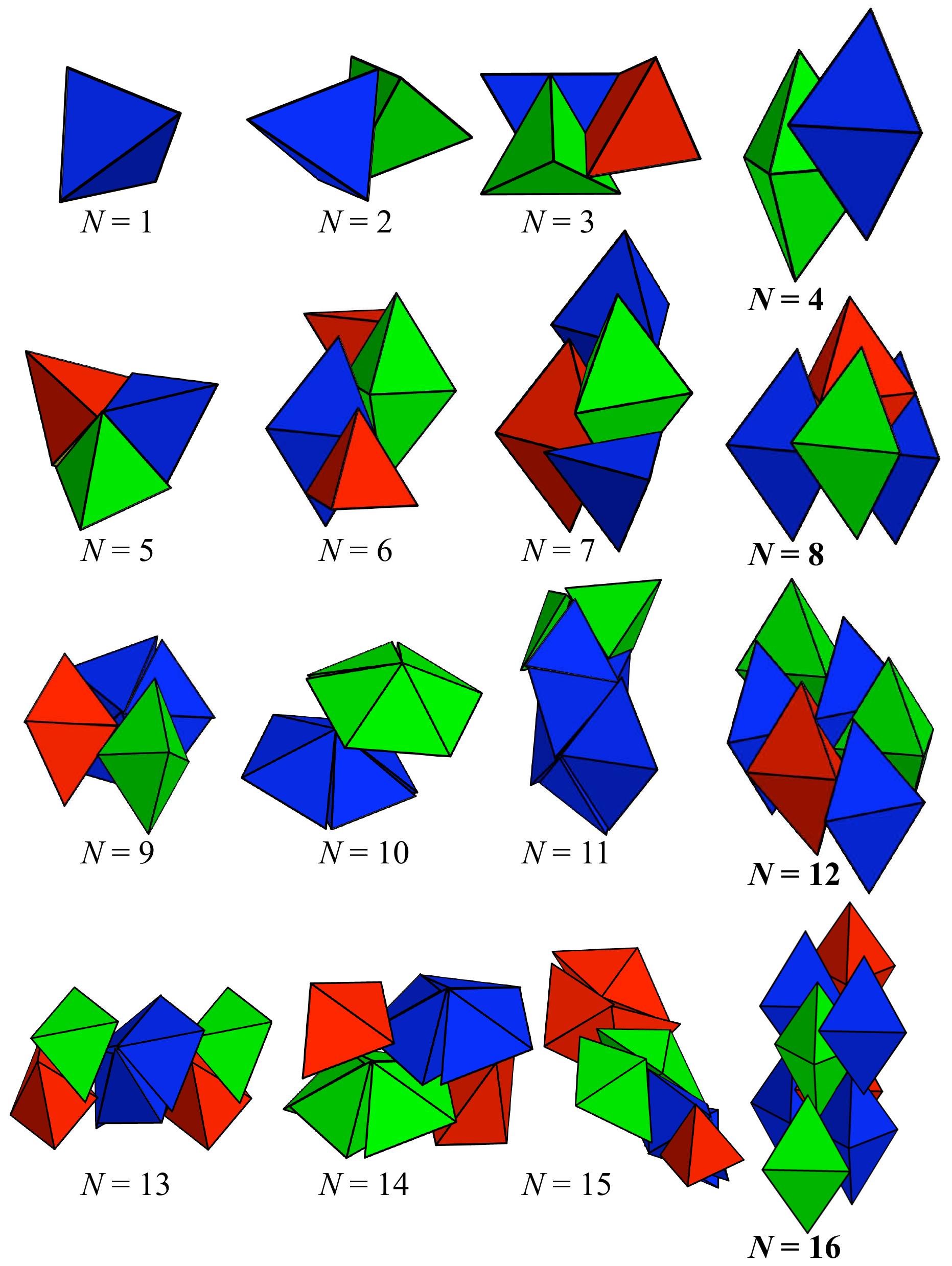}
\caption{(color online) Dense packings with up to 16 tetrahedra per cell as
  obtained from Monte Carlo simulations. The densest dimer packing is observed
  for $N=4,8,12,16$.}
\label{Fig:numfigs}
\end{figure}

For larger simulation cells crystalline packings are harder to achieve. The
$N=10$ packing consists of two pentagonal dipyramids, perfect in the sense
that four of each set of five tetrahedra are arranged face-to-face; the fifth
tetrahedron is oriented in such a way to distribute the (obligatory) gap
\cite{conway06} evenly on its two sides. In the cases $N=9$, 11, and 15 we
found no clear density gap separating very dense packings from less dense
packings. Either the density gap does not exist, or our simulations were not
successful in finding the optimal packings. The latter must be the case for
$N=15$, since a dimer packing with a vacancy can give a density of $\phi =
\tfrac{3750}{4671} = 0.802825\ldots$ Instead, the compressions became trapped
in disordered configurations with a network of pentagonal dipyramids, similar
to that found in previous simulations \cite{akbariNature09}. It is noteworthy
that all packings for $N\ge4$ achieve a density $\phi>\phi_{\rm sphere}$.

\section{Conclusions}

We have found the densest known packing of tetrahedra with density $\phi =
\tfrac{4000}{4671}$. This result was obtained as the optimal solution of a
three-parameter family of dimer packings, which is a generalization of
one-parameter and two-parameter families of packings recently reported with
lower maximum densities. Interestingly, the densest packing is not the most
symmetric one. Isobaric Monte Carlo simulations with variable cell shape
starting from random initial conditions recover the same high packing density
within $10^{-5}$ for small systems containing 4, 8, 12, and 16 tetrahedra.  We
also discovered new candidates for densest packings with 3, 5, 6, and 10
tetrahedra per unit cell. The analytical and numerical results combined
suggest that the packing density reported here could be the highest
achievable, at least for small $N$.

The dimer structures are remarkable in the relative simplicity of the
4-tetrahedron unit cell as compared to the 82-tetrahedron unit cell of the
quasicrystal approximant \cite{akbariNature09}, whose density is only slightly
less than that of the densest dimer packing. The dodecagonal quasicrystal is
the only ordered phase observed to form from random initial configurations of
large collections of tetrahedra at moderate densities. It is thus interesting
to note that for some certain values of $N$, when the small systems do not
form the dimer lattice packing, they instead prefer clusters (motifs) present
in the quasicrystal and its approximant, predominantly pentagonal dipyramids.
This suggests that the two types of packings -- the dimer crystal and the
quasicrystal/approximant -- may compete, raising interesting questions about
the relative stability of the two very different structures at finite
pressure. These questions will be explored in a forthcoming paper.

\section*{ACKNOWLEDGEMENTS}

E.R.C. is grateful to Jeffrey C. Lagarias for his many efforts to make this
research possible. We are very grateful to him and to the anonymous referees
for their helpful comments on the manuscript.

E.R.C. acknowledges a National Science Foundation RTG grant
DMS-0801029. M.E. acknowledges the support of a postdoctoral fellowship from
the Deutsche Forschungsgemeinschaft (EN 905/1-1). S.C.G. acknowledges support
from the Air Force Office of Scientific Research under MURI grant
FA9550-06-1-0337.

\eject


\begin{appendix}

\setlength{\parindent}{0pt}

\section{Pictures and equations}

This appendix contains figures of the double dimer packings (Figures A1-A3),
the three-dimensional restricted parameter space (Figure A4), and the volume
contour function (figure A5), as well as visualizations of the internal
symmetry of the restricted parameter space (Figures A6 and A7).

\medskip

{\bf Figure A0} introduces the orientation of the axes for all following
figures. We distinguish between the vector coordinate axes $\<x,y,z\>$ and the
packing coordinate axes $\<u,v,w\>$. In the former case, the viewpoints are
$\<1,1,1\>$, $\<0,-\cos{\pi\over 10},\sin{\pi\over 10}\>$, $\<0,0,1\>$, and
the ranges are $-{32\over 5}\le x,y,z\le +{32\over 5}$ (Figures A1-A3) or
$-{15\over 4}\le x,y,z\le+{15\over 4}$ (Figure A6). In the latter case, the
viewpoints are $\<\cos{2\pi\over 15}\cos{\pi\over 15},\cos{2\pi\over
  15}\sin{\pi\over 15},\sin{2\pi\over 15}\>$, $\<\sin{\pi\over
  120},0,\cos{\pi\over 120}\>$, and the range is $-{1\over 8}\le u,v,w\le
+{1\over 8}$ (Figures A4-A5 and A7).

\medskip

{\bf Figure A1} (4 pages) gives the notation and basic equations, and shows
the nearest intersecting neighbors. The colors of the equations match the
color of the respective neighbors. Since the double dimer packings are cluster
transitive on dimers, we consider only the dimer at the origin. Its positive
and negative neighbors are shown for three pairs of special packings at the
optimal points, the central points, and the symmetric points in parameter
space (Figure A4). The pairs of packings are related by the symmetry $T$
(Figure A7). Notice that the sets of negative neighbors are always the same,
whereas the sets of positive neighbors are different. The neighbors for the
optimal packings are related by a reflection through the plane $+2x-y-z = 0$,
the neighbors for the central packings by a reflection through the plane
$+2x-y-z = 0$, and the neighbors for the symmetric packings by a reflection
through the plane which contains $\pm a\pm b$.

\medskip

{\bf Figure A2} (2 pages) lists all 20 possible neighbors that can intersect
with the dimer at the origin together with their linear incidence conditions.
The entire figure is a floor plan diagram of the \emph{neighbors}, ordered by
layer. Notice how the positive layers and the negative layers complement each
other. The even layers are related by symmetry, which is broken for the odd
layers due to the presence of a small translation by the vector $d$.

\medskip

{\bf Figure A3} (2 pages) shows larger portions of the layers for the optimal
packings.  The entire figure is a floor plan diagram of the \emph{packing},
ordered by layer. Notice the almost triangular/hexagonal planar basis
mentioned in Remarks 2.  We show both optimal packings for comparison to
demonstrate the subtle differences. The symmetry $T$ that maps the packings onto
each other is a two-fold rotation around an axis running from bottom left to
top right in the left columns (view along the 'almost' three-fold axis).

\medskip

{\bf Figure A4} (1 page) depicts the restricted three-dimensional parameter
space of dense double dimer packings $\<u,v,w\>$ together with special planes,
lines, and points. The shape is a strongly deformed tetrahedron extended along
the $u$-axis. It is bounded by four \emph{boundary planes}. Packings with
maximal (minimal) density are located on the \emph{maximal (minimal) lines}
(Figure A5). The optimal plane contains the two \emph{optimal points}, the
origin, and the \emph{symmetric line}.

\medskip

{\bf Figure A5} (1 page) visualizes contours of the unit cell volume $V =
{9\over 25}(117+60u_{}^2-80uv-80v_{}^2)$ in parameter space. The lattice
volume is a hyperbolic paraboloid in $\<u,v,V\>$.  The contours are
hyperbolic cylinders in $\<u,v,w\>$ parallel to the $w$-axis. Their
intersection with the restricted parameter space (gray) allows to identify the
locations of volume extrema on the maximal and minimal lines.

\medskip

{\bf Figure A6} (2 pages) demonstrates the distortion in the geometry of
nearest neighbors in vector space $\<x,y,z\>$. The colored spheres (colors
match the colors in Figures A1-A3) are positioned at the centers of the twelve
dimers in the layer 0, $\pm1$, closest to the dimer at the origin. Lines
connect dimers that are adjacent in the dimer lattice. A deformation of the
neighbor shell occurs when moving along the various lines in parameter space
$\<u,v,w\>$.  This deformation is illustrated in each sub-figure by overlaying
three neighbor configurations along these lines.  Packings along the central
line have the symmetry $T$, which is a two-fold rotation about the axis $p =
\<+2,-1,-1\>$ (see proof of Theorem 2). Notice how these packings and
distortions look more symmetric than the others. Packings along the maximal
lines form a double family. Packings along the minimal lines form two single
families that are distorted in complementary ways.

\medskip

{\bf Figure A7} (1 page) illustrates the effect of the symmetry $T$ on the
restricted parameter space shown in Figure A4 by identifying the equivalent
packings under this symmetry: $T\<u,v,w\> = \<-u,-v,w\>$. Within contour
planes of the form $w = *$, i.e. normal to the $w$-axis, $T$ acts as a
combination of a shear and a two-fold rotation around the symmetric line
$\<0,0,w\>$. Equivalent points are colored identically.  Along the minimal
lines, opposite points on the same line are related. Notice that each minimal
line is contained in a single contour plane.  Along the maximal lines,
opposite points on opposite lines are related.  Notice that the two optimal
points are contained in a single contour plane. The final sub-figure shows the
change of the contour plane intersection with the maximal lines for a sequence
of five contour planes.

{\centering
\bigskip\includegraphics[width=\linewidth]{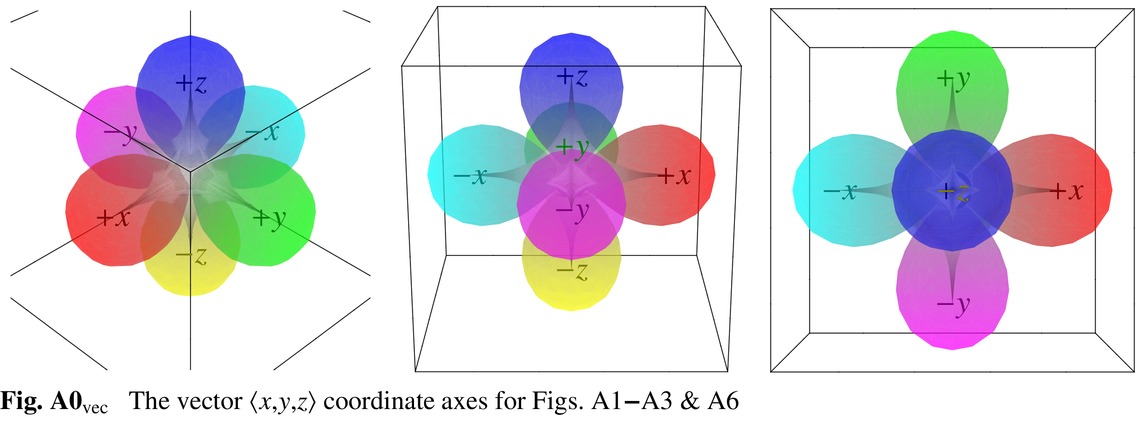}
\bigskip\includegraphics[width=\linewidth]{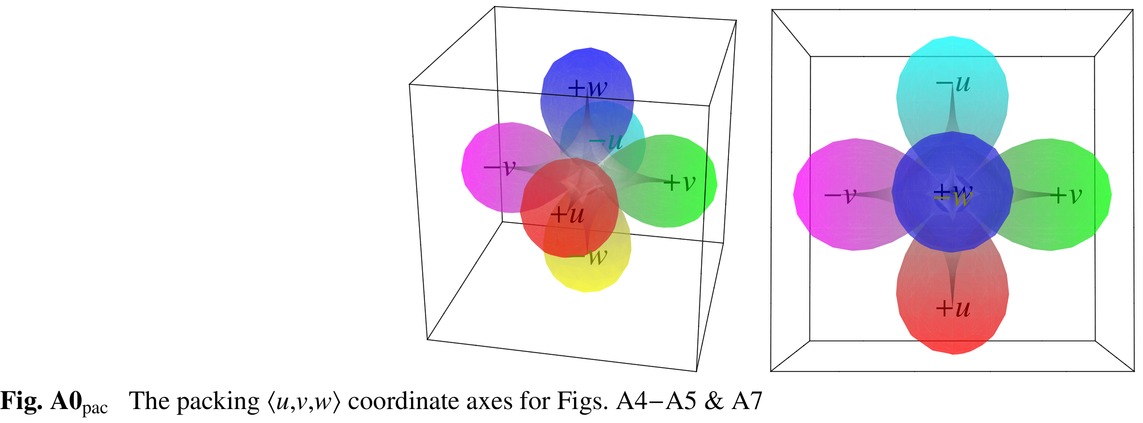}
\eject\includegraphics[width=0.85\linewidth]{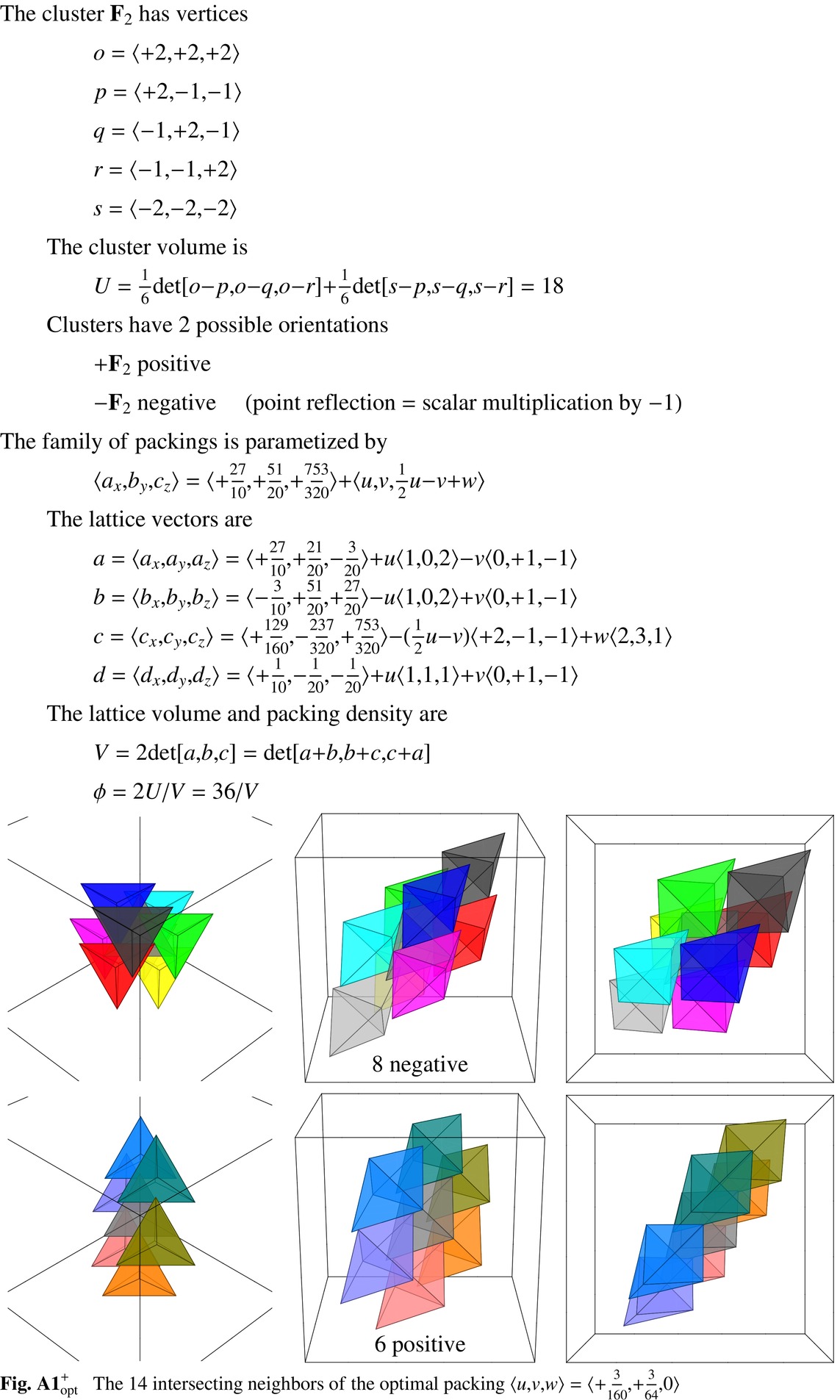}
\eject\includegraphics[width=0.85\linewidth]{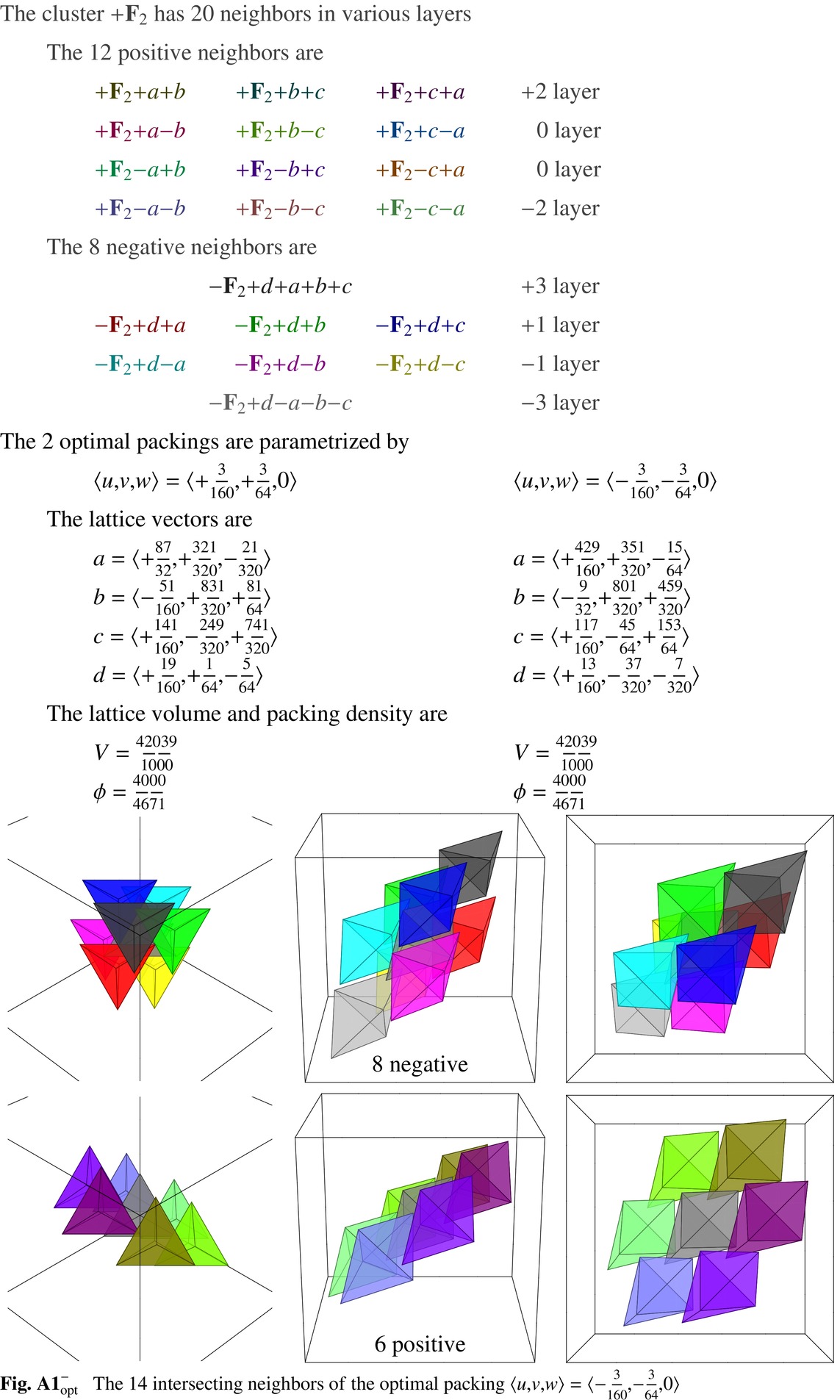}
\eject\includegraphics[width=0.95\linewidth]{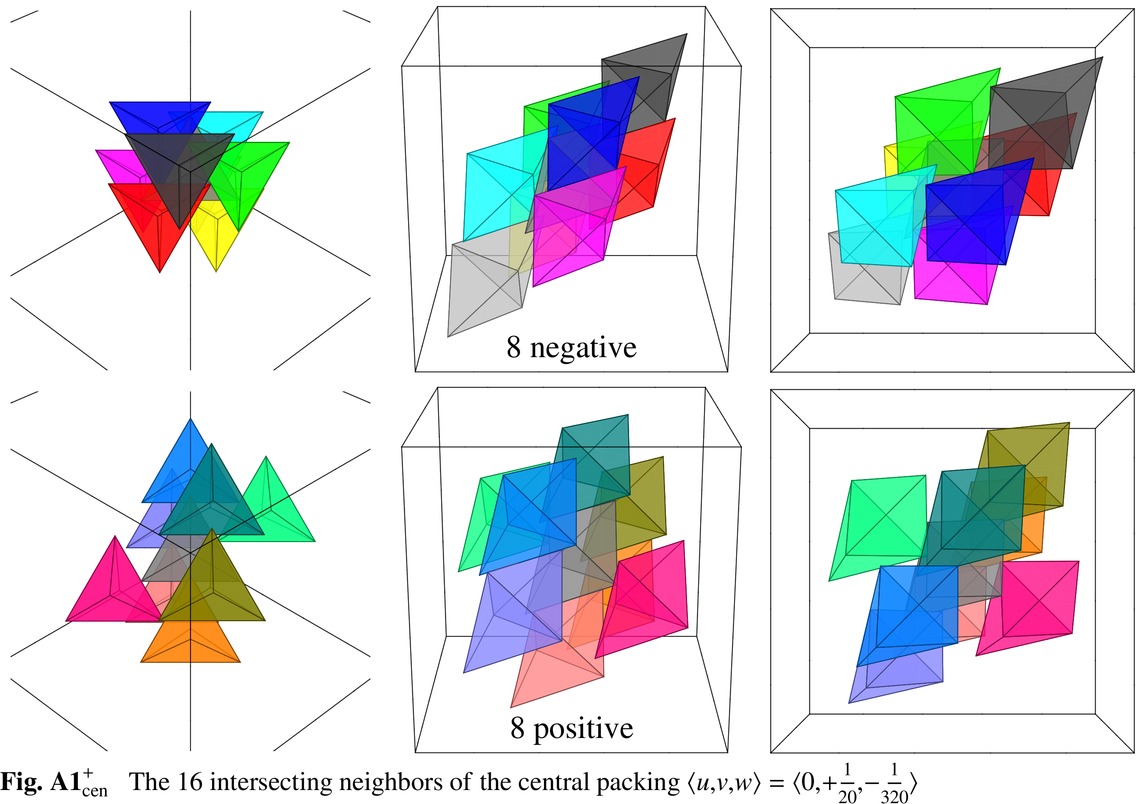}
\vskip .5in\includegraphics[width=0.95\linewidth]{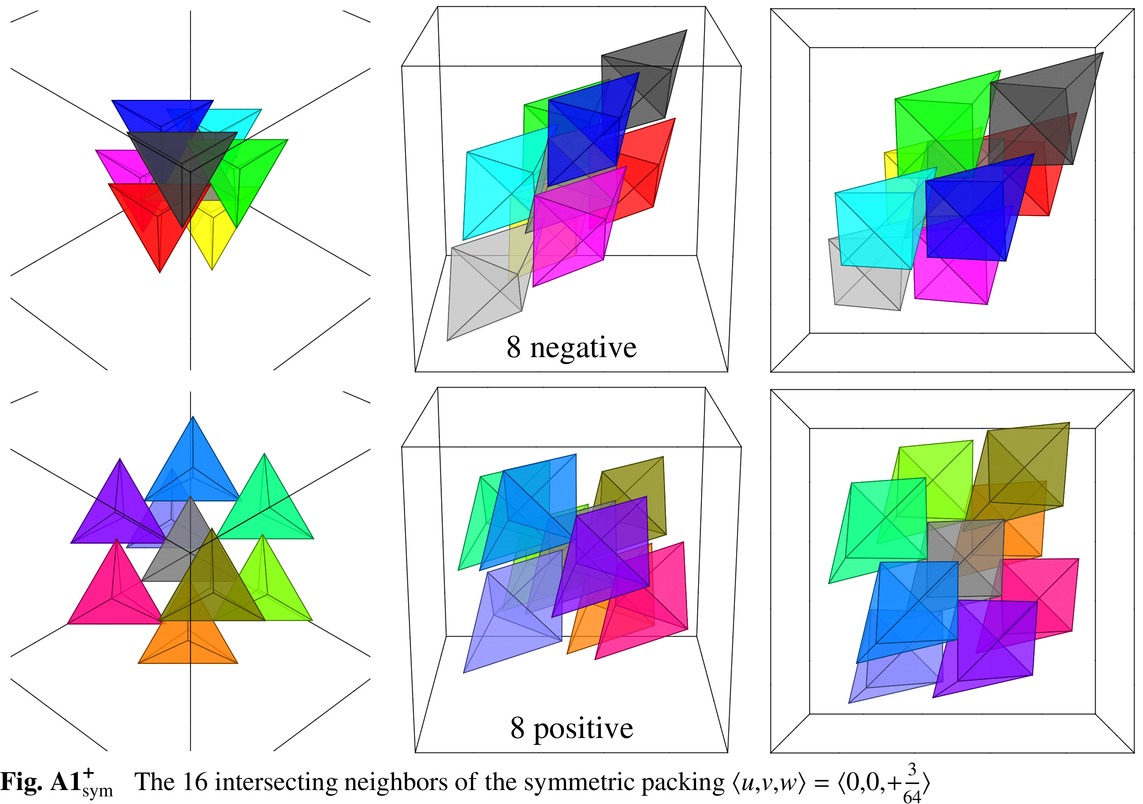}
\eject\includegraphics[width=0.95\linewidth]{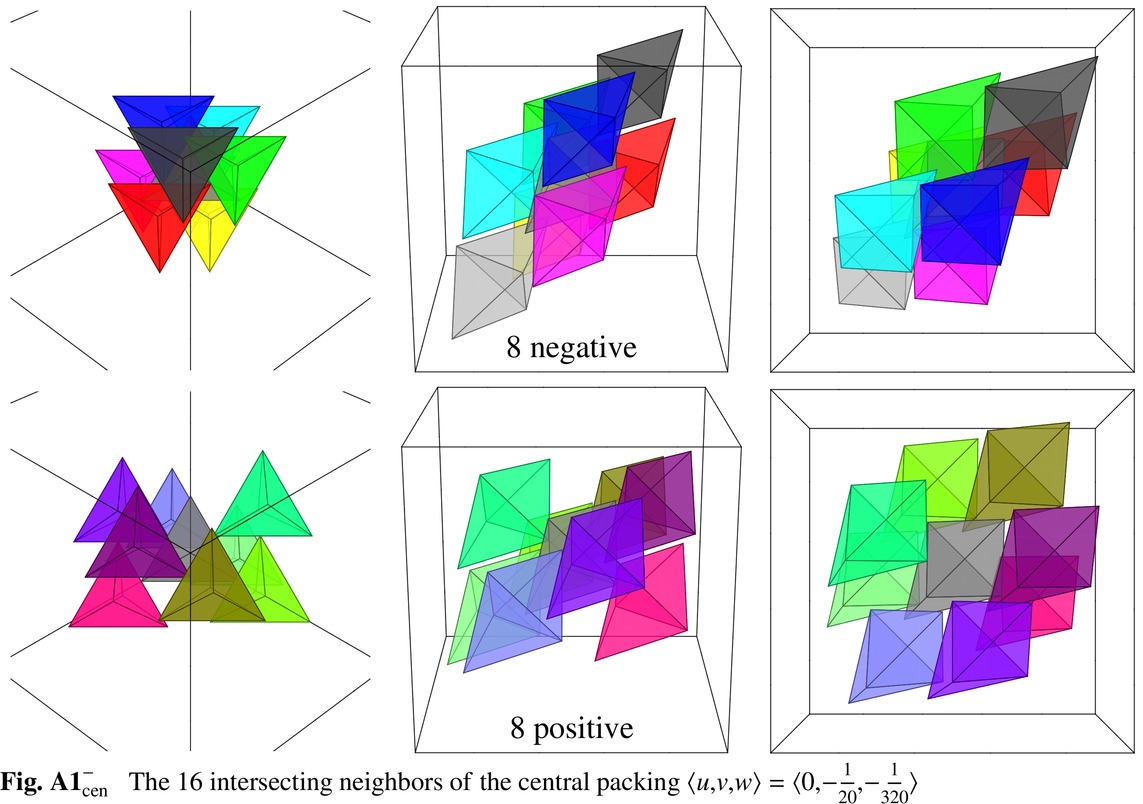}
\vskip .5in\includegraphics[width=0.95\linewidth]{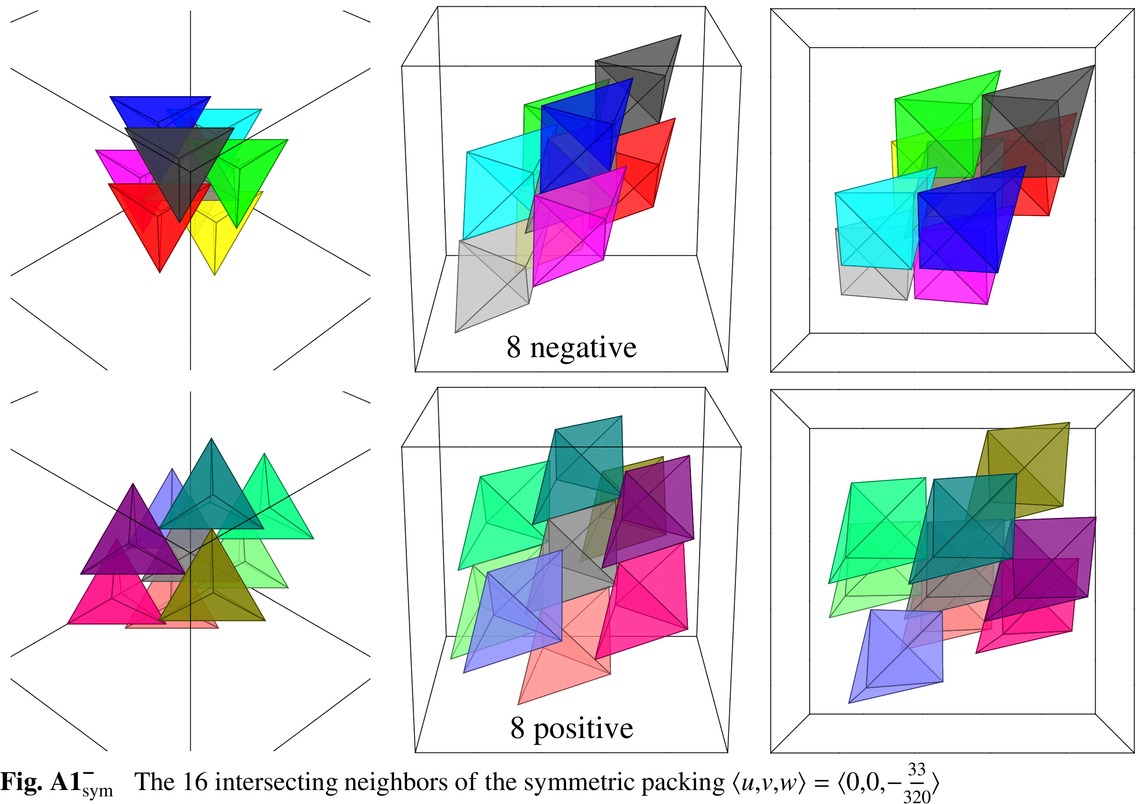}
\eject\includegraphics[width=0.80\linewidth]{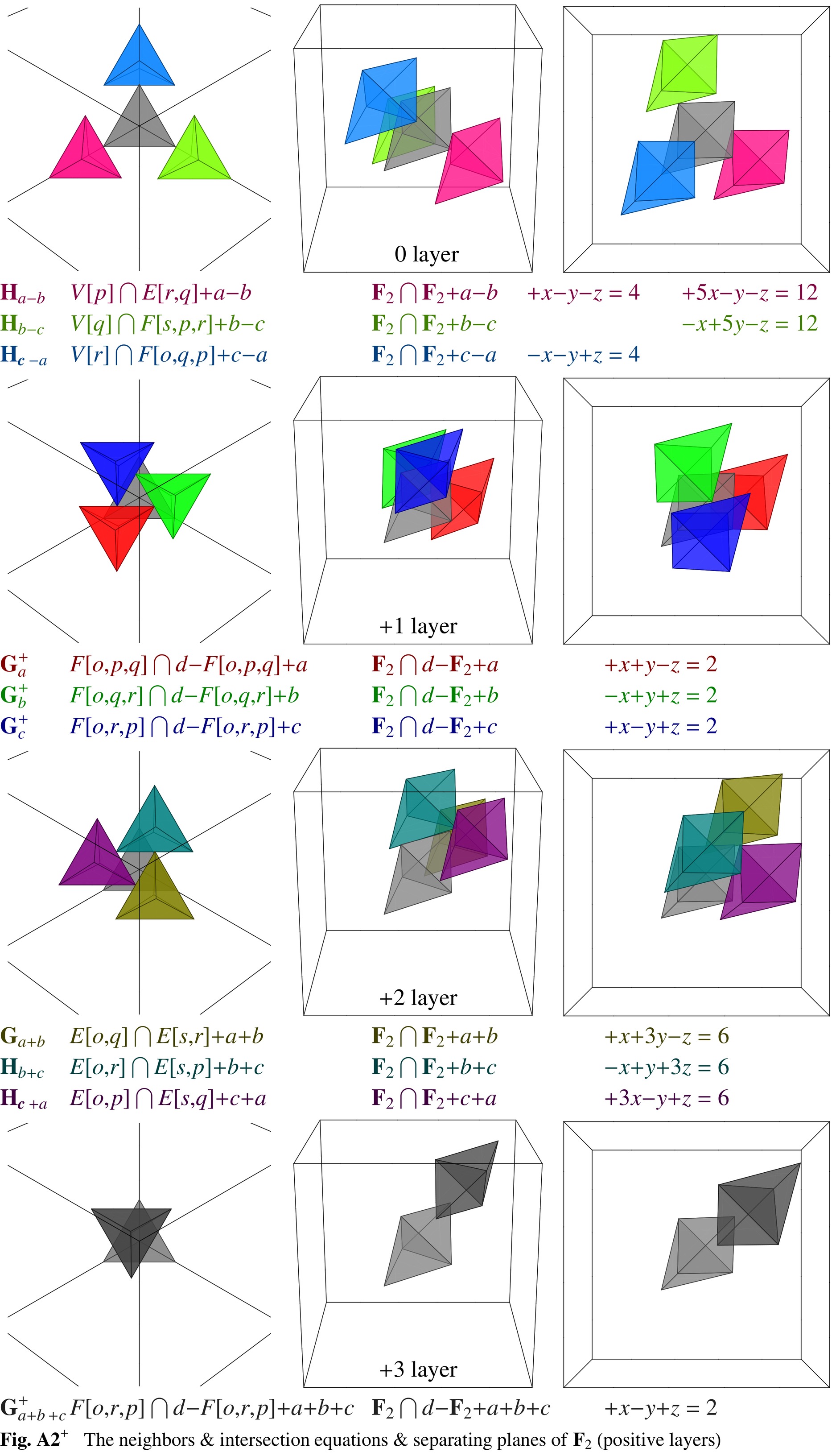}
\eject\includegraphics[width=0.80\linewidth]{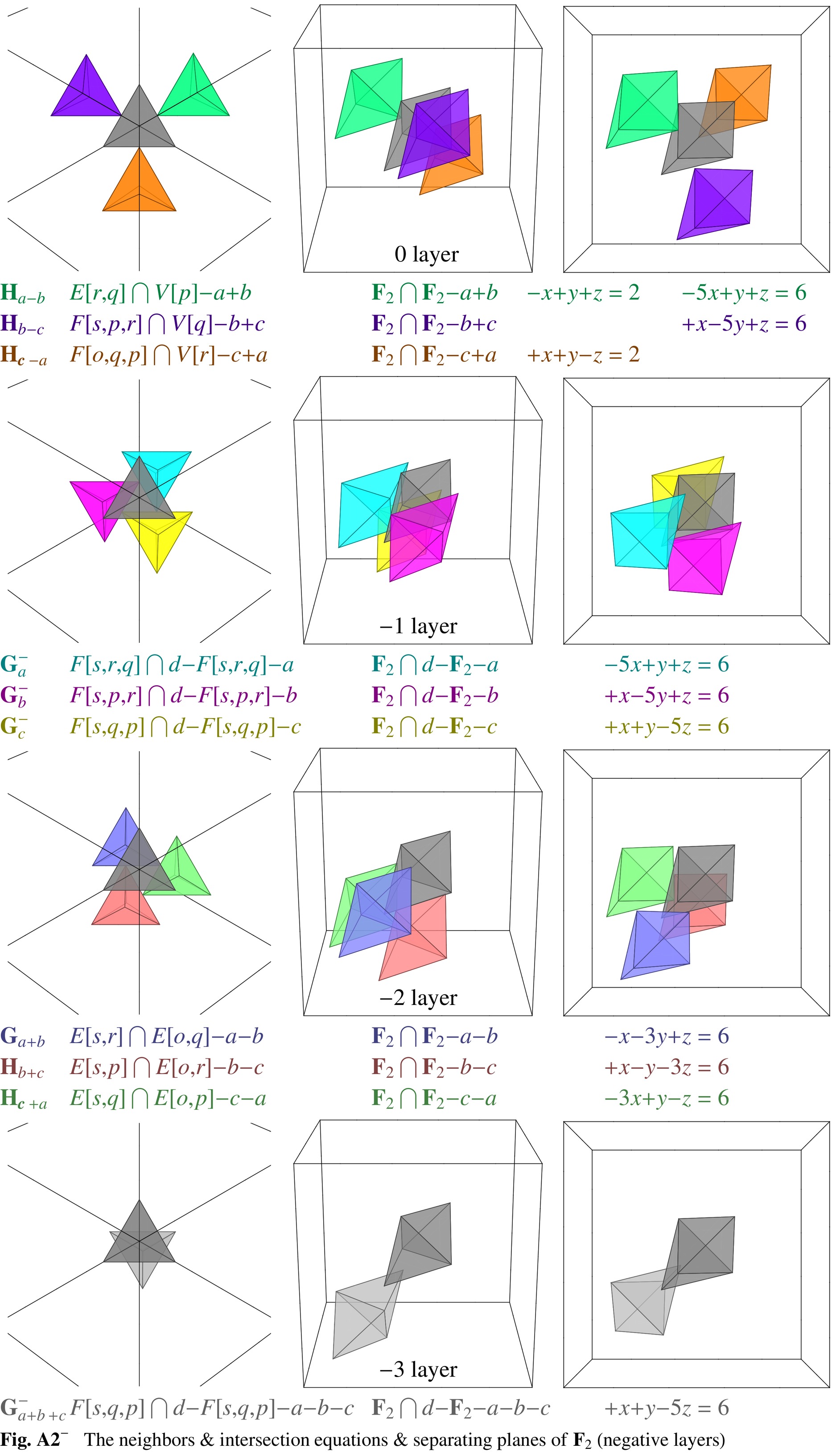}
\eject\includegraphics[width=1.00\linewidth]{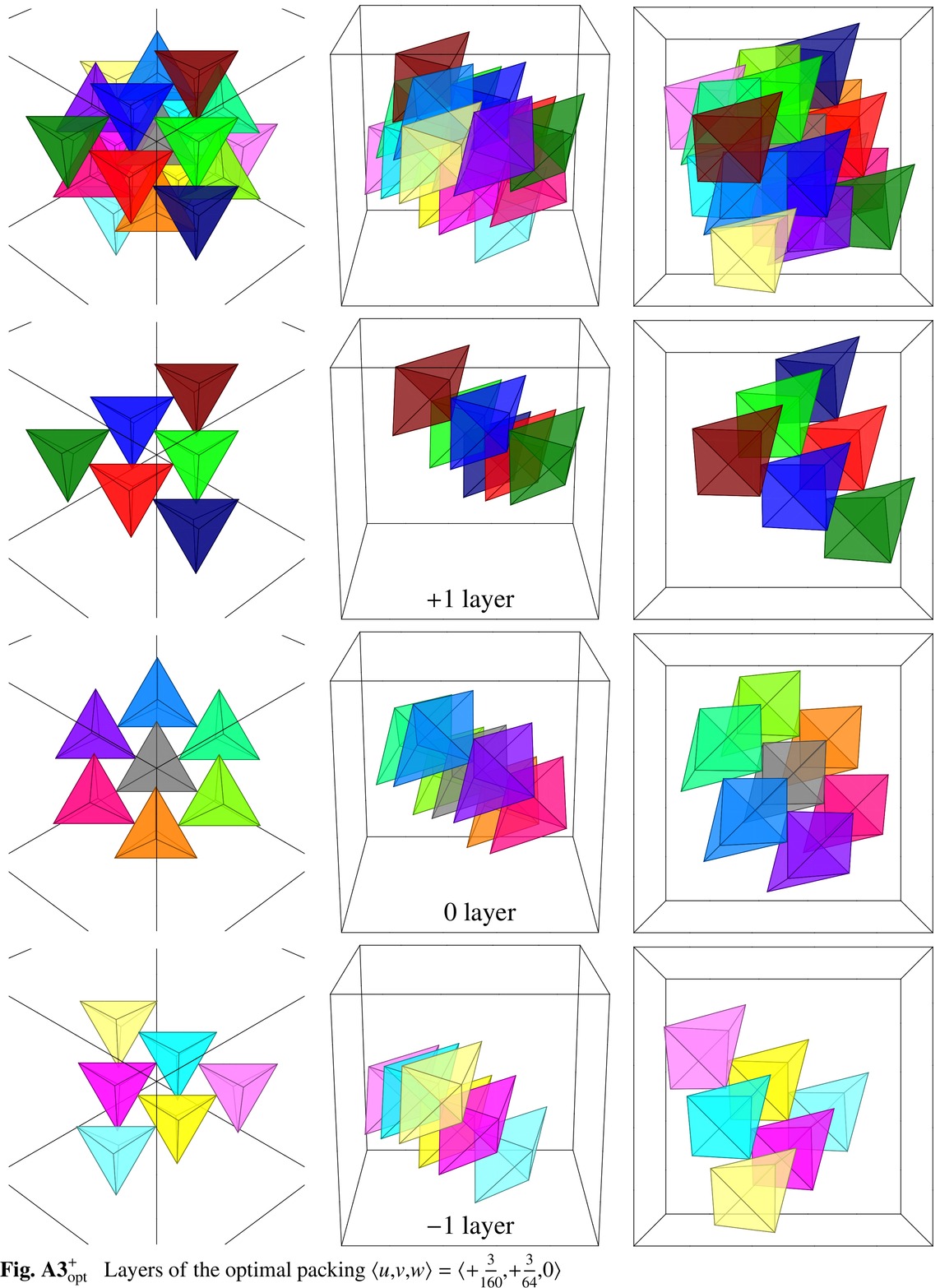}
\eject\includegraphics[width=1.00\linewidth]{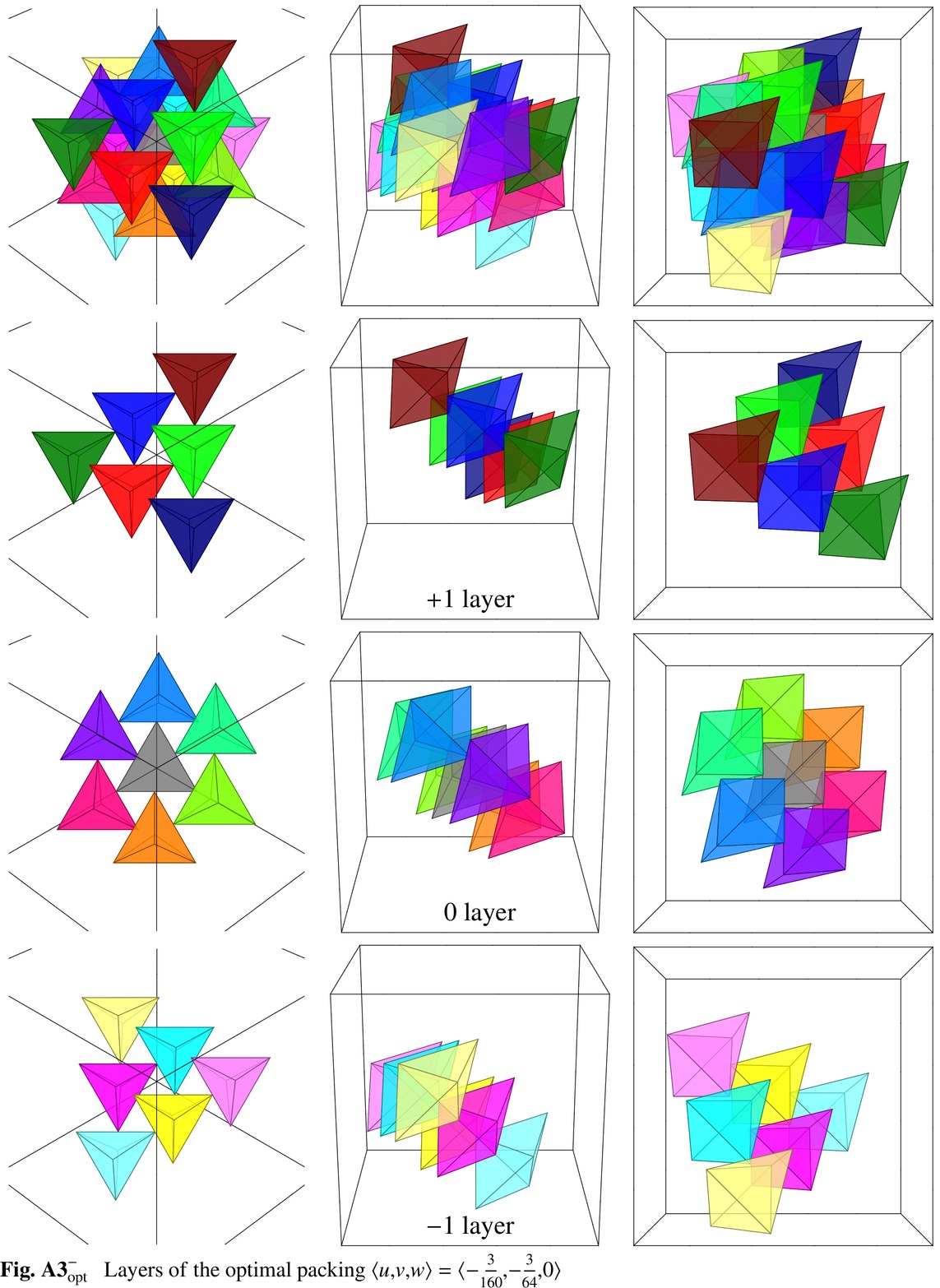}
\eject\includegraphics[width=0.83\linewidth]{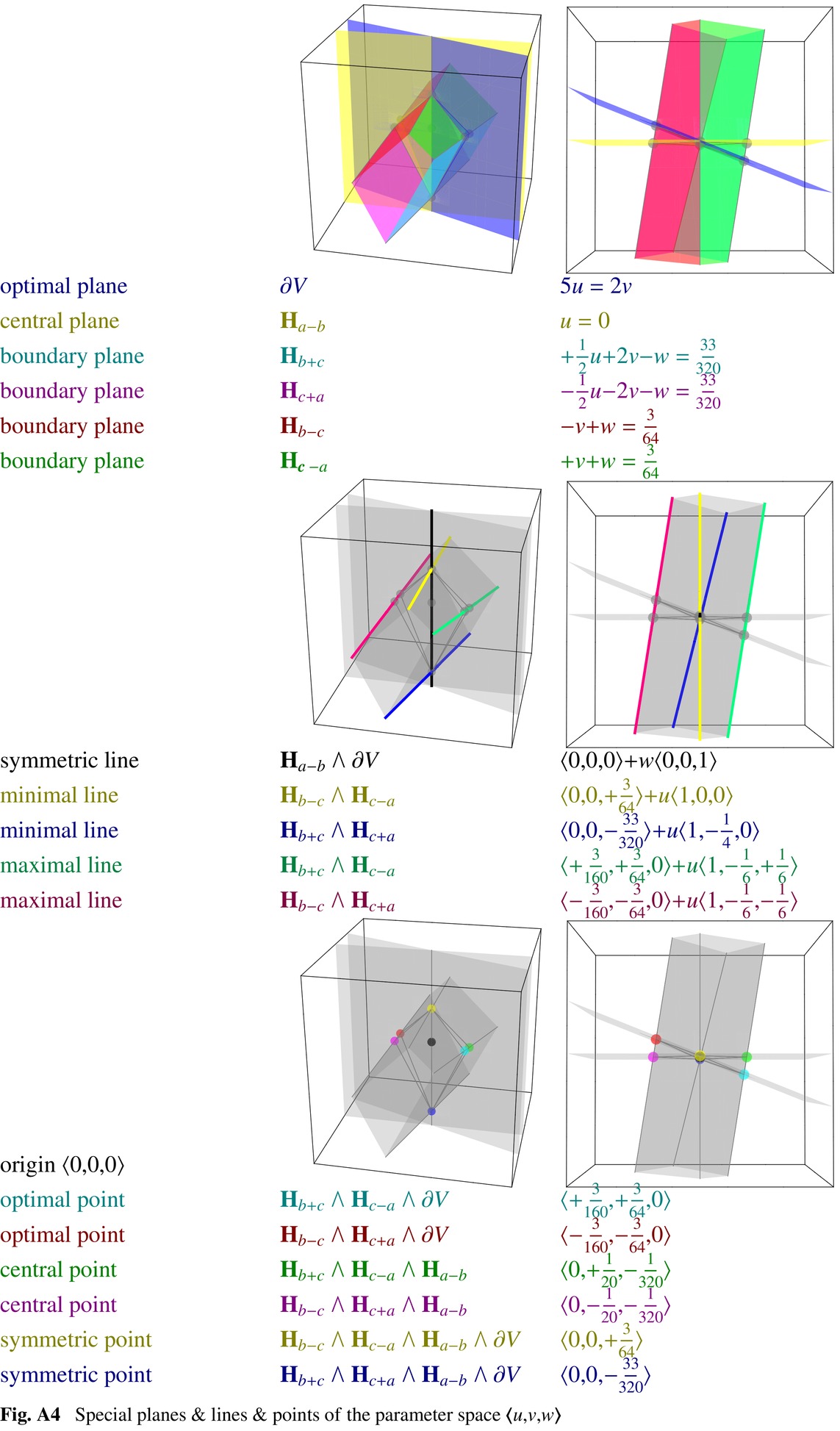}
\eject\includegraphics[width=0.83\linewidth]{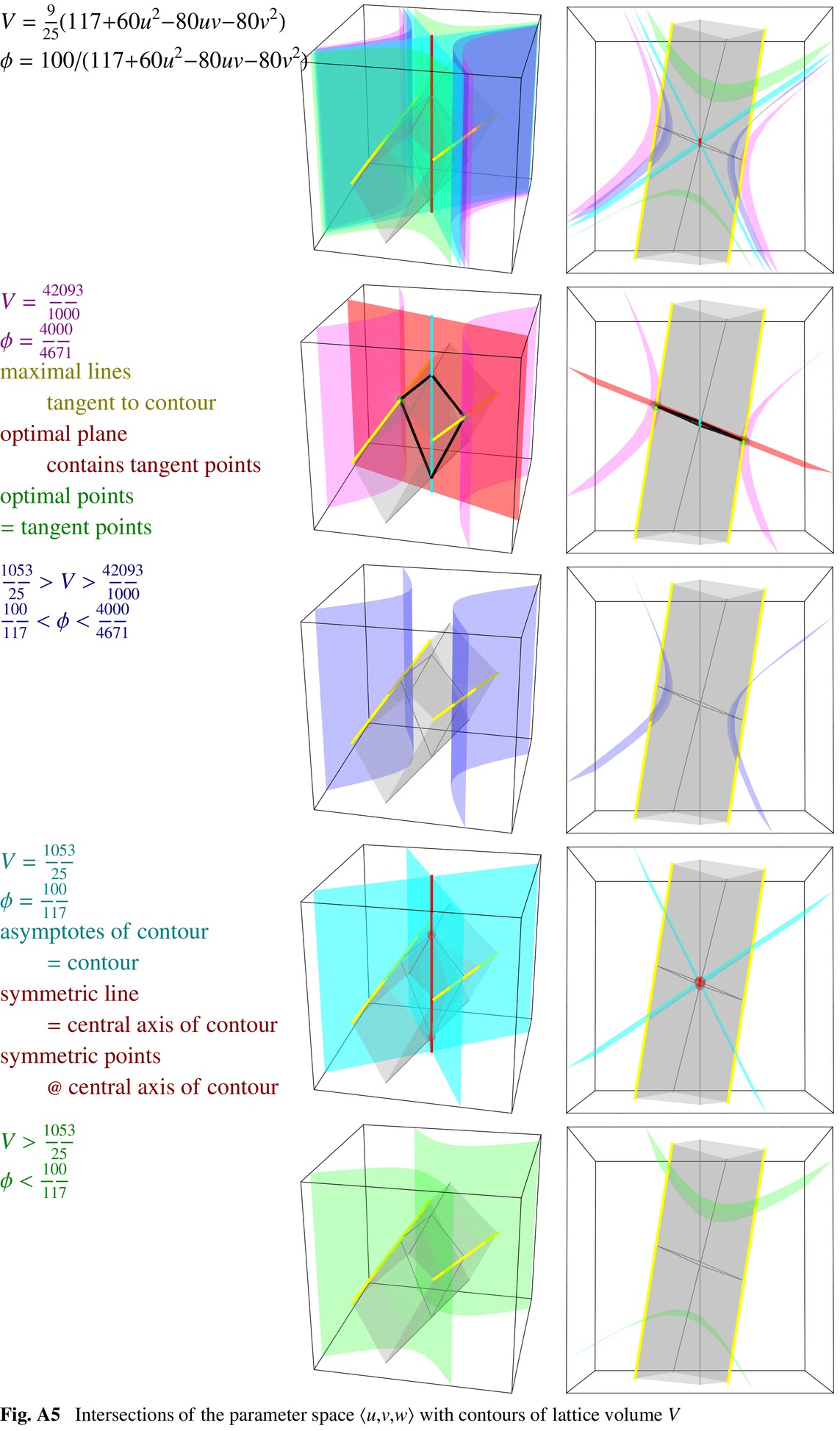}
\eject\includegraphics[width=0.90\linewidth]{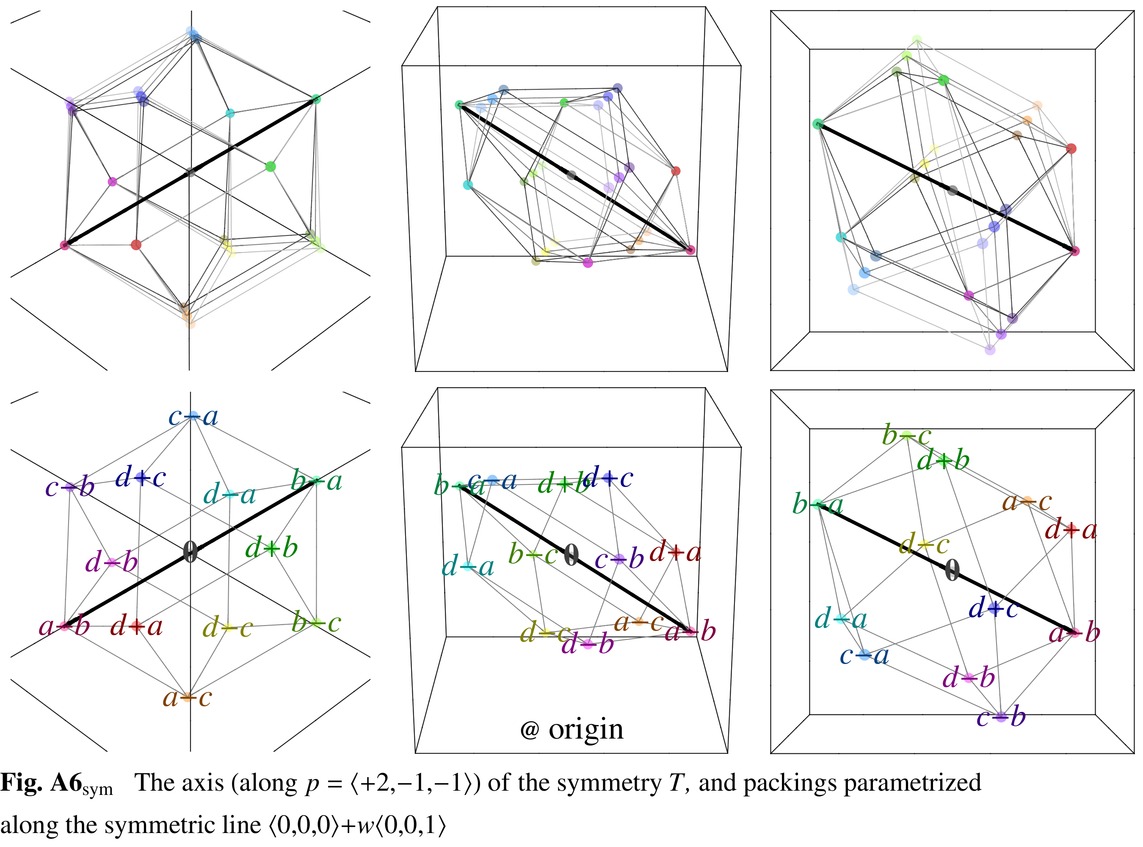}
\vskip .5in\includegraphics[width=0.90\linewidth]{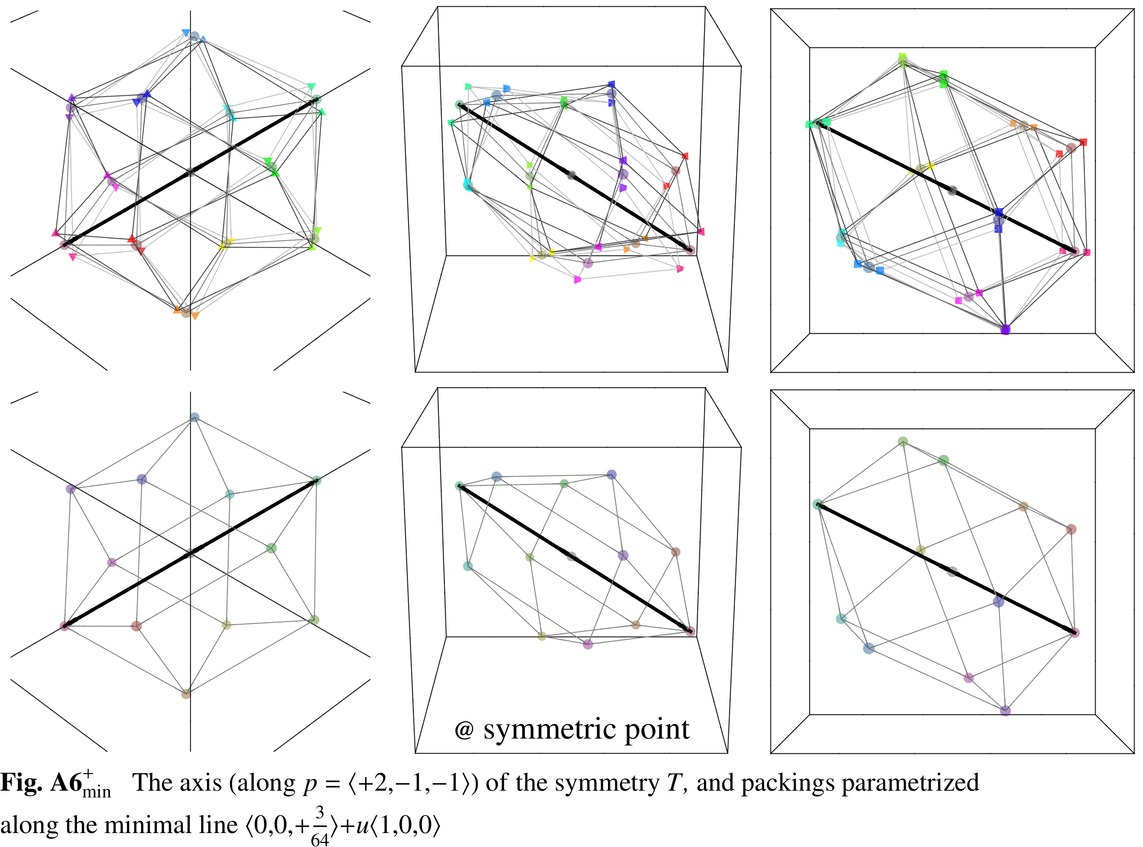}
\eject\includegraphics[width=0.90\linewidth]{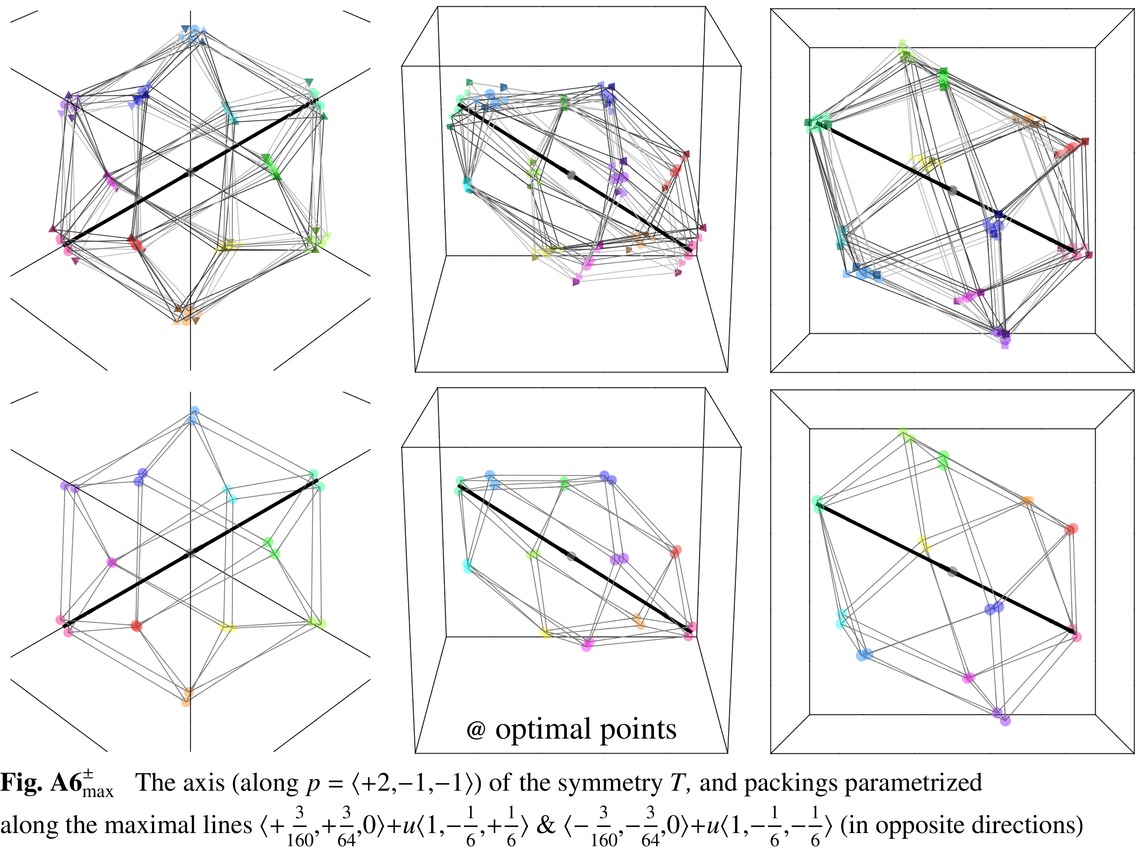}
\vskip .5in\includegraphics[width=0.90\linewidth]{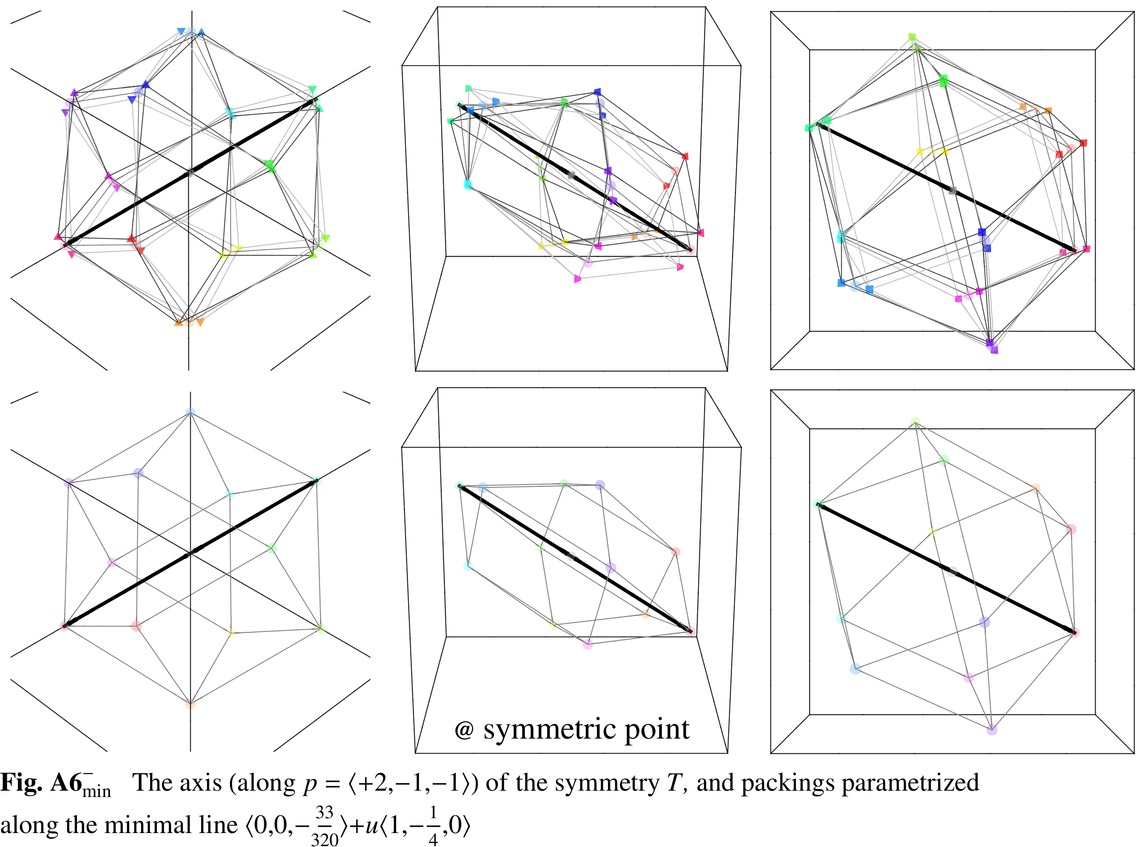}
\eject\includegraphics[width=0.83\linewidth]{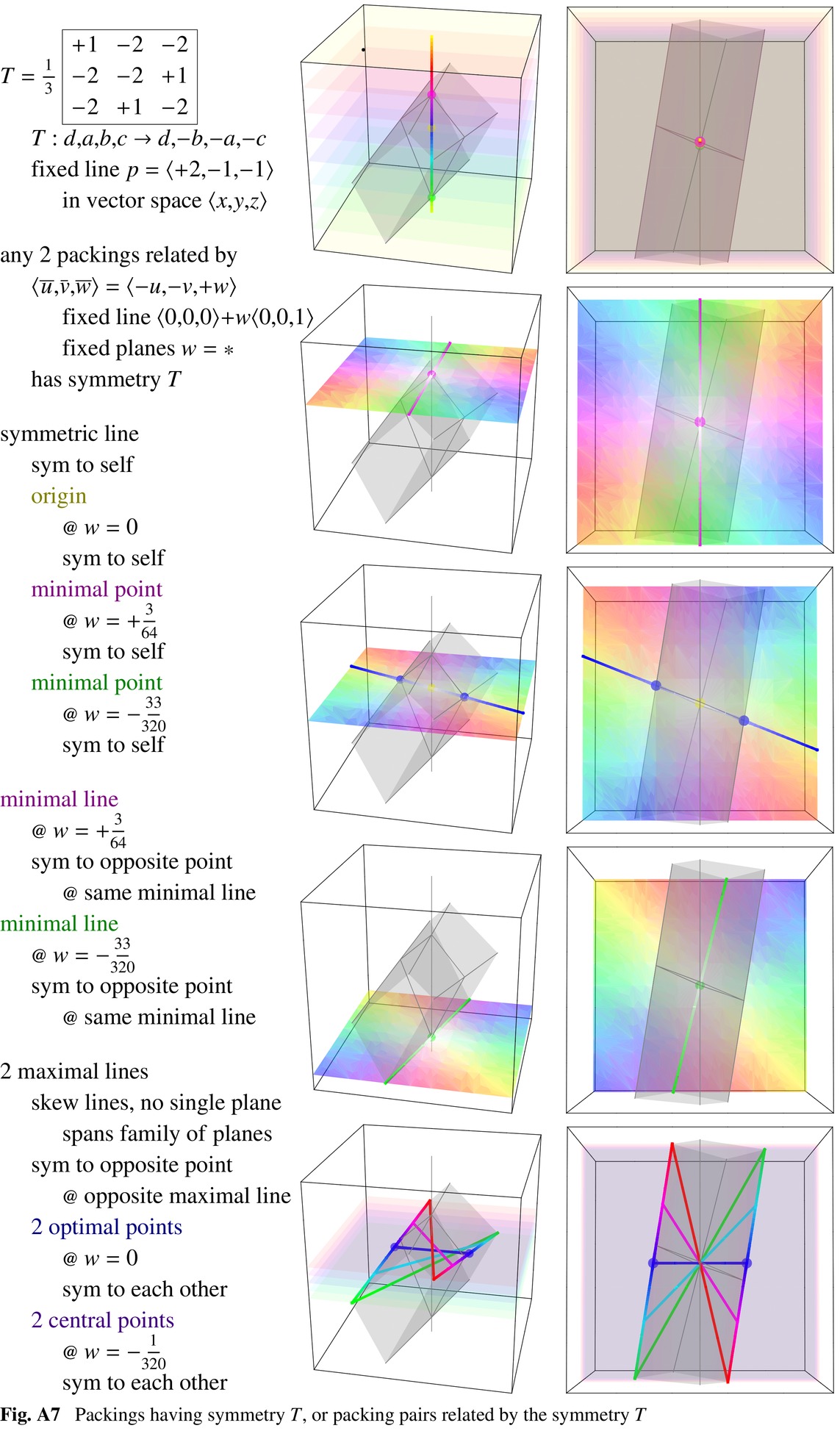}
}


\newcommand{\matrixold}[1]{\begin{matrix}#1\end{matrix}}

\section{Symmetry}

\medskip{\bf B0\tab Origin $\<0,0,0\>$}

\medskip\tabb The symmetry group is ${\bf Z}_2^2$

\medskip\tabb\tabb $T_0^{} = \left[\scriptsize\matrixold{1&0&0\cr 0&1&0\cr 0&0&1}\right] = T_0^{-1}$

\medskip\tabb\tabb\tab $T_0^{}\,:\,{\rm packing}\,\<0,0,0\>\,\rightarrow\,{\rm packing}\,\<0,0,0\>$
\par\tabb\tabb\tab $\hskip 33pt d,a,b,c\hskip 37pt d,+a,+b,+c$
\par\tabb\tabb\tab $\hskip 28.5pt o,p,q,r,s\hskip 43pt o,p,q,r,s$

\medskip\tabb\tabb $T_c^{} = {1\over 624}\left[\scriptsize\matrixold{+538&+172&-344\cr +79&+466&+316\cr -251&+502&-380}\right] = T_c^{-1}$

\medskip\tabb\tabb\tab $T_c^{}\,:\,{\rm packing}\,\<0,0,0\>\,\rightarrow\,{\rm packing}\,\<0,0,0\>$
\par\tabb\tabb\tab $\hskip 33pt d,a,b,c\hskip 37pt d,+a,+b,-c$

\medskip\tabb\tabb $T_{ab}^{} = {1\over 208}\left[\scriptsize\matrixold{+98&-196&-24\cr -165&-86&-36\cr -55&-98&+196}\right] = T_{ab}^{-1}$

\medskip\tabb\tabb\tab $T_{ab}^{}\,:\,{\rm packing}\,\<0,0,0\>\,\rightarrow\,{\rm packing}\,\<0,0,0\>$
\par\tabb\tabb\tab $\hskip 37pt d,a,b,c\hskip 37pt d,-a,-b,+c$

\medskip\tabb\tabb $T_{abc}^{} = {1\over 3}\left[\scriptsize\matrixold{+1&-2&-2\cr -2&-2&+1\cr -2&+1&-2}\right] = T_{abc}^{-1}$\tabb rotation by $\pm\pi$ about $p = \<+2,-1,-1\>$

\medskip\tabb\tabb\tab $T_{abc}^{}\,:\,{\rm packing}\,\<0,0,0\>\,\rightarrow\,{\rm packing}\,\<0,0,0\>$
\par\tabb\tabb\tab $\hskip 40pt d,a,b,c\hskip 37pt d,-a,-b,-c$
\par\tabb\tabb\tab $\hskip 35.5pt o,p,q,r,s\hskip 43pt s,p,r,q,o$

\medskip{\bf B1\tab Symmetric line $\<0,0,0\>+w\<0,0,1\>$}

\medskip\tabb The symmetry group is ${\bf Z}$

\medskip\tabb\tabb $S_{(w)}^{} = \left[\scriptsize\matrixold{1&0&0\cr 0&1&0\cr 0&0&1}\right]+{10\over 117}w\left[\scriptsize\matrixold{+2&-4&+8\cr +3&-6&+12\cr +1&-2&+4}\right] = S_{(-w)}^{-1}$

\medskip\tabb\tabb\tab $S_{(w)}^{}\,:\,{\rm packing}\,\<0,0,0\>\,\rightarrow\,{\rm packing}\,\<0,0,+w\>$
\par\tabb\tabb\tab $\hskip 41pt d,a,b,c\hskip 53pt d,a,b,c$

\medskip\tabb\tabb $S_{(w)}^{-1} = \left[\scriptsize\matrixold{1&0&0\cr 0&1&0\cr 0&0&1}\right]-{10\over 117}w\left[\scriptsize\matrixold{+2&-4&+8\cr +3&-6&+12\cr +1&-2&+4}\right] = S_{(-w)}^{}$

\medskip\tabb\tabb\tab $S_{(w)}^{-1}\,:\,{\rm packing}\,\<0,0,0\>\,\rightarrow\,{\rm packing}\,\<0,0,-w\>$
\par\tabb\tabb\tab $\hskip 42pt d,a,b,c\hskip 53pt d,a,b,c$

\eject{\bf B2\tab Symmetric line $\<0,0,0\>+w\<0,0,1\>$}

\medskip\tabb The symmetry group is ${\bf Z}_2^2$

\medskip\tabb\tabb $S_{(w)}^{}T_0^{}S_{(w)}^{-1} = T_0^{}$

\medskip\tabb\tabb\tab $S_{(w)}^{}T_0^{}S_{(w)}^{-1}\,:\,{\rm packing}\,\<0,0,w\>\,\rightarrow\,{\rm packing}\,\<0,0,w\>$
\par\tabb\tabb\tab $\hskip 67pt d,a,b,c\hskip 38pt d,+a,+b,+c$
\par\tabb\tabb\tab $\hskip 62.5pt o,p,q,r,s\hskip 44pt o,p,q,r,s$

\medskip\tabb\tabb $S_{(w)}^{}T_c^{}S_{(w)}^{-1} = T_c^{}-{20\over 117}w\left[\scriptsize\matrixold{+2&-4&+8\cr +3&-6&+12\cr +1&-2&+4}\right]$

\medskip\tabb\tabb\tab $S_{(w)}^{}T_c^{}S_{(w)}^{-1}\,:\,{\rm packing}\,\<0,0,w\>\,\rightarrow\,{\rm packing}\,\<0,0,w\>$
\par\tabb\tabb\tab $\hskip 67pt d,a,b,c\hskip 38pt d,+a,+b,-c$

\medskip\tabb\tabb $S_{(w)}^{}T_{ab}^{}S_{(w)}^{-1} = T_{ab}^{}+{20\over 117}w\left[\scriptsize\matrixold{+2&-4&+8\cr +3&-6&+12\cr +1&-2&+4}\right]$

\medskip\tabb\tabb\tab $S_{(w)}^{}T_{ab}^{}S_{(w)}^{-1}\,:\,{\rm packing}\,\<0,0,w\>\,\rightarrow\,{\rm packing}\,\<0,0,w\>$
\par\tabb\tabb\tab $\hskip 71pt d,a,b,c\hskip 38pt d,-a,-b,+c$

\medskip\tabb\tabb $S_{(w)}^{}T_{abc}^{}S_{(w)}^{-1} = T_{abc}^{}$\tabb rotation by $\pm\pi$ about $p = \<+2,-1,-1\>$

\medskip\tabb\tabb\tab $S_{(w)}^{}T_{abc}^{}S_{(w)}^{-1}\,:\,{\rm packing}\,\<0,0,w\>\,\rightarrow\,{\rm packing}\,\<0,0,w\>$
\par\tabb\tabb\tab $\hskip 74pt d,a,b,c\hskip 38pt d,-a,-b,-c$
\par\tabb\tabb\tab $\hskip 69.5pt o,p,q,r,s\hskip 44pt s,p,r,q,o$

\medskip{\bf B3\tab Parameter space}

\medskip\tabb For any packing $\<u,v,w\>$

\medskip\tabb\tabb $R_a^{} = R_a^{-1}\,:\,\<x,y,z\>\,\rightarrow\,d-\<x,y,z\>+a$\tabb inversion about ${1\over 2}(d+a)$

\medskip\tabb\tabb\tab $R_a^{}\,:\,{\rm packing}\,\<u,v,w\>\,\rightarrow\,{\rm packing}\,\<u,v,w\>$
\par\tabb\tabb\tab $\hskip 35pt d,a,b,c\hskip 34pt a,d,d+a-b,d-c+a$

\medskip\tabb\tabb $R_b^{} = R_b^{-1}\,:\,\<x,y,z\>\,\rightarrow\,d-\<x,y,z\>+b$\tabb inversion about ${1\over 2}(d+b)$

\medskip\tabb\tabb\tab $R_b^{}\,:\,{\rm packing}\,\<u,v,w\>\,\rightarrow\,{\rm packing}\,\<u,v,w\>$
\par\tabb\tabb\tab $\hskip 35pt d,b,c,a\hskip 34pt b,d,d+b-c,d-a+b$

\medskip\tabb\tabb $R_c^{} = R_c^{-1}\,:\,\<x,y,z\>\,\rightarrow\,d-\<x,y,z\>+c$\tabb inversion about ${1\over 2}(d+c)$

\medskip\tabb\tabb\tab $R_c^{}\,:\,{\rm packing}\,\<u,v,w\>\,\rightarrow\,{\rm packing}\,\<u,v,w\>$
\par\tabb\tabb\tab $\hskip 35pt d,c,a,b\hskip 34pt c,d,d+c-a,d-b+c$

\medskip\tabb For any 2 packings $\<u,v,w\>$, $\<\tilde u,\tilde v,\tilde w\>$ related by

\medskip\tabb\tabb\tab $\<\tilde u,\tilde v,-{1\over 2}\tilde u+\tilde v+\tilde w\> = \<-u,-v,-{1\over 2}u+v+w\>$

\medskip\tabb\tabb $T_{abc}^{} = {1\over 3}\left[\scriptsize\matrixold{+1&-2&-2\cr -2&-2&+1\cr -2&+1&-2}\right]$\tabb rotation by $\pm\pi$ about $p = \<+2,-1,-1\>$

\medskip\tabb\tabb\tab $T_{abc}^{}\,:\,{\rm packing}\,\<u,v,w\>\,\rightarrow\,{\rm packing}\,\<\tilde u,\tilde v,\tilde w\>$
\par\tabb\tabb\tab $\hskip 41pt d,a,b,c\hskip 37pt d,-a,-b,-c$
\par\tabb\tabb\tab $\hskip 36.5pt o,p,q,r,s\hskip 43pt s,p,r,q,o$

\section{Computations}

\medskip{\bf C0$_{\rm vec}^{}$\tab Vector space $\<x,y,z\>$}

\medskip\tabb The cluster vertices are

\medskip\tabb\tabb $o = \<+2,+2,+2\>$
\par\tabb\tabb $p = \<+2,-1,-1\>$
\par\tabb\tabb $q = \<-1,+2,-1\>$
\par\tabb\tabb $r = \<-1,-1,+2\>$
\par\tabb\tabb $s = \<-2,-2,-2\>$

\medskip\tabb The cluster volume is

\medskip\tabb\tabb $U = {1\over 6}\det[o-p,o-q,o-r]+{1\over 6}\det[s-p,s-q,s-r] = 18$

\medskip\tabb The lattice vectors are

\medskip\tabb\tabb $a = \<a_x^{},a_y^{},a_z^{}\>$
\par\tabb\tabb $b = \<b_x^{},b_y^{},b_z^{}\>$
\par\tabb\tabb $c = \<c_x^{},c_y^{},c_z^{}\>$
\par\tabb\tabb $d = \<d_x^{},d_y^{},d_z^{}\>$

\medskip\tabb The lattice volume and packing density are

\medskip\tabb\tabb $V = 2\det[a,b,c] = \det[a+b,b+c,c+a]$
\par\tabb\tabb $\phi = 2U/V = 36/V$

\medskip{\bf C0$_{\rm pac}^{}$\tab Packing parameter space $\<u,v,w\>$}

\medskip\tabb The linear incidence conditions are

\medskip\tabb\tabb ${\bf G}_a^\pm,{\bf G}_b^\pm,{\bf G}_c^\pm,{\bf G}_{abc}^\pm,{\bf G}_{a+b}^{}$
\par\tabb\tabb $\<a_x^{},b_y^{},c_z^{}\> = \<u+{27\over 10},v+{51\over 20},w+{753\over 320}\>$

\medskip\tabb The lattice vectors are

\medskip\tabb\tabb $a = \<{27\over 10}+u,{21\over 20}-v,-{3\over 20}+2u+v\>$
\par\tabb\tabb $b = \<-{3\over 10}-u,{51\over 20}+v,{27\over 20}-2u-v\>$
\par\tabb\tabb $c = \<{129\over 160}-2u+4v+2w,-{237\over 320}-u+2v+3w,{753\over 320}+w\>$
\par\tabb\tabb $d = \<{1\over 10}+u,-{1\over 20}+u+v,-{1\over 20}+u-v\>$

\medskip\tabb The lattice volume and packing density are

\medskip\tabb\tabb $V = 2\det[a,b,c] = {9\over 25}(117+60u_{}^2-80uv-80v_{}^2)$
\par\tabb\tabb $\phi = 36/V = 100/(117+60u_{}^2-80uv-80v_{}^2)$

\medskip{\bf C1$_{\partial V}^{}$/C1$_{\rm opt}^{}$\tab Optimal plane $u = {2\over 5}v$}

\medskip\tabb The lattice volume function is a rotated parabolic hyperboloid

\medskip\tabb\tabb $V = {9\over 25}(117+60u_{}^2-80uv-80v_{}^2)$

\eject\tabb The contours (of constant $V$) are rotated hyperbolas

\medskip\tabb\tabb ${3\over 5}u_{}^2-{4\over 5}uv-{4\over 5}v_{}^2 = {1\over 36}V-{117\over 100}$

\medskip\tabb The asymptotes (of the hyperbolas) are

\medskip\tabb\tabb $u = +2v$
\par\tabb\tabb $u = -{2\over 3}v$

\medskip\tabb The 2 maximal lines have slope $du/dv = -6$

\medskip\tabb\tabb $\<+{3\over 160},+{3\over 64},-{3\over 80}\>+u\<1,-{1\over 6},+{5\over 6}\>$
\par\tabb\tabb $\<-{3\over 160},-{3\over 64},+{3\over 80}\>+u\<1,-{1\over 6},+{1\over 2}\>$

\medskip\tabb The contours have slope $du/dv = -6$ when

\medskip\tabb\tabb $u = +{2\over 5}v$

\medskip{\bf C1$_{a-b}^{}$/C1$_{\rm cen}^{}$\tab Central plane $u = 0$}

\medskip\tabb The linear incidence conditions are

\medskip\tabb\tabb ${\bf G}_a^\pm,{\bf G}_b^\pm,{\bf G}_c^\pm,{\bf G}_{abc}^\pm,{\bf G}_{a+b}^{},{\bf H}_{a-b}^{}$
\par\tabb\tabb $\<b_y^{},c_z^{}\> = \<v+{51\over 20},w+{753\over 320}\>$

\medskip\tabb The lattice vectors are

\medskip\tabb\tabb $a = \<{27\over 10},{21\over 20}-v,-{3\over 20}+v\>$
\par\tabb\tabb $b = \<-{3\over 10},{51\over 20}+v,{27\over 20}-v\>$
\par\tabb\tabb $c = \<{129\over 160}+4v+2w,-{237\over 320}+2v+3w,{753\over 320}+w\>$
\par\tabb\tabb $d = \<{1\over 10},-{1\over 20}+v,-{1\over 20}-v\>$

\medskip\tabb The lattice volume and packing density are

\medskip\tabb\tabb $V = {9\over 25}(117-80v_{}^2)$
\par\tabb\tabb $\phi = 100/(117-80v_{}^2)$

\medskip{\bf C1$_{b+c}^{}$\tab Boundary plane $+u+v-w = {33\over 320}$}

\medskip\tabb The linear incidence conditions are

\medskip\tabb\tabb ${\bf G}_a^\pm,{\bf G}_b^\pm,{\bf G}_c^\pm,{\bf G}_{abc}^\pm,{\bf G}_{a+b}^{},{\bf H}_{b+c}^{}$
\par\tabb\tabb $\<a_x^{},b_y^{}\> = \<u+{27\over 10},v+{51\over 20}\>$

\medskip\tabb The lattice vectors are

\medskip\tabb\tabb $a = \<{27\over 10}+u,{21\over 20}-v,-{3\over 20}+2u+v\>$
\par\tabb\tabb $b = \<-{3\over 10}-u,{51\over 20}+v,{27\over 20}-2u-v\>$
\par\tabb\tabb $c = \<{3\over 5}+6v,-{21\over 20}+2u+5v,{9\over 4}+u+v\>$
\par\tabb\tabb $d = \<{1\over 10}+u,-{1\over 20}+u+v,-{1\over 20}+u-v\>$

\medskip\tabb The lattice volume and packing density are

\medskip\tabb\tabb $V = {9\over 25}(117+60u_{}^2-80uv-80v_{}^2)$
\par\tabb\tabb $\phi = 100/(117+60u_{}^2-80uv-80v_{}^2)$

\eject{\bf C1$_{c+a}^{}$\tab Boundary plane $-3v-w = {33\over 320}$}

\medskip\tabb The linear incidence conditions are

\medskip\tabb\tabb ${\bf G}_a^\pm,{\bf G}_b^\pm,{\bf G}_c^\pm,{\bf G}_{abc}^\pm,{\bf G}_{a+b}^{},{\bf H}_{c+a}^{}$
\par\tabb\tabb $\<a_x^{},b_y^{}\> = \<u+{27\over 10},v+{51\over 20}\>$

\medskip\tabb The lattice vectors are

\medskip\tabb\tabb $a = \<{27\over 10}+u,{21\over 20}-v,-{3\over 20}+2u+v\>$
\par\tabb\tabb $b = \<-{3\over 10}-u,{51\over 20}+v,{27\over 20}-2u-v\>$
\par\tabb\tabb $c = \<{3\over 5}-2u-2v,-{21\over 20}-u-7v,{9\over 4}-3v\>$
\par\tabb\tabb $d = \<{1\over 10}+u,-{1\over 20}+u+v,-{1\over 20}+u-v\>$

\medskip\tabb The lattice volume and packing density are

\medskip\tabb\tabb $V = {9\over 25}(117+60u_{}^2-80uv-80v_{}^2)$
\par\tabb\tabb $\phi = 100/(117+60u_{}^2-80uv-80v_{}^2)$

\medskip{\bf C1$_{b-c}^{}$\tab Boundary plane $-u+2w = {3\over 32}$}

\medskip\tabb The linear incidence conditions are

\medskip\tabb\tabb ${\bf G}_a^\pm,{\bf G}_b^\pm,{\bf G}_c^\pm,{\bf G}_{abc}^\pm,{\bf G}_{a+b}^{},{\bf H}_{b-c}^{}$
\par\tabb\tabb $\<a_x^{},b_y^{}\> = \<u+{27\over 10},v+{51\over 20}\>$

\medskip\tabb The lattice vectors are

\medskip\tabb\tabb $a = \<{27\over 10}+u,{21\over 20}-v,-{3\over 20}+2u+v\>$
\par\tabb\tabb $b = \<-{3\over 10}-u,{51\over 20}+v,{27\over 20}-2u-v\>$
\par\tabb\tabb $c = \<{9\over 10}-u+4v,-{3\over 5}+{1\over 2}u+2v,{12\over 5}+{1\over 2}u\>$
\par\tabb\tabb $d = \<{1\over 10}+u,-{1\over 20}+u+v,-{1\over 20}+u-v\>$

\medskip\tabb The lattice volume and packing density are

\medskip\tabb\tabb $V = {9\over 25}(117+60u_{}^2-80uv-80v_{}^2)$
\par\tabb\tabb $\phi = 100/(117+60u_{}^2-80uv-80v_{}^2)$

\medskip{\bf C1$_{c-a}^{}$\tab Boundary plane $-u+4v+2w = {3\over 32}$}

\medskip\tabb The linear incidence conditions are

\medskip\tabb\tabb ${\bf G}_a^\pm,{\bf G}_b^\pm,{\bf G}_c^\pm,{\bf G}_{abc}^\pm,{\bf G}_{a+b}^{},{\bf H}_{c-a}^{}$
\par\tabb\tabb $\<a_x^{},b_y^{}\> = \<u+{27\over 10},v+{51\over 20}\>$

\medskip\tabb The lattice vectors are

\medskip\tabb\tabb $a = \<{27\over 10}+u,{21\over 20}-v,-{3\over 20}+2u+v\>$
\par\tabb\tabb $b = \<-{3\over 10}-u,{51\over 20}+v,{27\over 20}-2u-v\>$
\par\tabb\tabb $c = \<{9\over 10}-u,-{3\over 5}+{1\over 2}u-4v,{12\over 5}+{1\over 2}u-2v\>$
\par\tabb\tabb $d = \<{1\over 10}+u,-{1\over 20}+u+v,-{1\over 20}+u-v\>$

\medskip\tabb The lattice volume and packing density are

\medskip\tabb\tabb $V = {9\over 25}(117+60u_{}^2-80uv-80v_{}^2)$
\par\tabb\tabb $\phi = 100/(117+60u_{}^2-80uv-80v_{}^2)$

\eject{\bf C2$_{\rm sym}^{}$\tab Symmetric line $\<0,0,0\>+w\<0,0,1\>$}

\medskip\tabb The linear incidence conditions are

\medskip\tabb\tabb ${\bf G}_a^\pm,{\bf G}_b^\pm,{\bf G}_c^\pm,{\bf G}_{abc}^\pm,{\bf G}_{a+b}^{},{\bf H}_{a-b}^{},\partial V$
\par\tabb\tabb $\<b_y^{},c_z^{}\> = \<v+{51\over 20},w+{753\over 320}\>,v = 0$ [min $V$]

\medskip\tabb The lattice vectors are

\medskip\tabb\tabb $a = \<{27\over 10},{21\over 20},-{3\over 20}\>$
\par\tabb\tabb $b = \<-{3\over 10},{51\over 20},{27\over 20}\>$
\par\tabb\tabb $c = \<{129\over 160}+2w,-{237\over 320}+3w,{753\over 320}+w\>$
\par\tabb\tabb $d = \<{1\over 10},-{1\over 20},-{1\over 20}\>$

\medskip\tabb The lattice volume and packing density are

\medskip\tabb\tabb $V = {1053\over 25}$
\par\tabb\tabb $\phi = {100\over 117}$

\medskip{\bf C2$_{\rm min}^+$\tab Minimal line $\<0,0,+{3\over 64}\>+u\<1,0,+{1\over 2}\>$}

\medskip\tabb The linear incidence conditions are

\medskip\tabb\tabb ${\bf G}_a^\pm,{\bf G}_b^\pm,{\bf G}_c^\pm,{\bf G}_{abc}^\pm,{\bf G}_{a+b}^{},{\bf H}_{b-c}^{},{\bf H}_{c-a}^{}$
\par\tabb\tabb $a_x^{} = u+{27\over 10}$

\medskip\tabb The lattice vectors are

\medskip\tabb\tabb $a = \<{27\over 10}+u,{21\over 20},-{3\over 20}+2u\>$
\par\tabb\tabb $b = \<-{3\over 10}-u,{51\over 20},{27\over 20}-2u\>$
\par\tabb\tabb $c = \<{9\over 10}-u,-{3\over 5}+{1\over 2}u,{12\over 5}+{1\over 2}u\>$
\par\tabb\tabb $d = \<{1\over 10}+u,-{1\over 20}+u,-{1\over 20}+u\>$

\medskip\tabb The lattice volume and packing density are

\medskip\tabb\tabb $V = {9\over 25}(117+60u_{}^2)$
\par\tabb\tabb $\phi = 100/(117+60u_{}^2)$

\medskip{\bf C2$_{\rm min}^-$\tab Minimal line $\<0,0,-{33\over 320}\>+u\<1,-{1\over 4},+{3\over 4}\>$}

\medskip\tabb The linear incidence conditions are

\medskip\tabb\tabb ${\bf G}_a^\pm,{\bf G}_b^\pm,{\bf G}_c^\pm,{\bf G}_{abc}^\pm,{\bf G}_{a+b}^{},{\bf H}_{b+c}^{},{\bf H}_{c+a}^{}$
\par\tabb\tabb $a_x^{} = u+{27\over 10}$

\medskip\tabb The lattice vectors are

\medskip\tabb\tabb $a = \<{27\over 10}+u,{21\over 20}+{1\over 4}u,-{3\over 20}+{7\over 4}u\>$
\par\tabb\tabb $b = \<-{3\over 10}-u,{51\over 20}-{1\over 4}u,{27\over 20}-{7\over 4}u\>$
\par\tabb\tabb $c = \<{3\over 5}-{3\over 2}u,-{21\over 20}+{3\over 4}u,{9\over 4}+{3\over 4}u\>$
\par\tabb\tabb $d = \<{1\over 10}+u,-{1\over 20}+{3\over 4}u,-{1\over 20}+{5\over 4}u\>$

\medskip\tabb The lattice volume and packing density are

\medskip\tabb\tabb $V = {9\over 25}(117+75u_{}^2)$
\par\tabb\tabb $\phi = 100/(117+75u_{}^2)$

\eject{\bf C2$_{\rm max}^+$\tab Maximal line $\<0,+{1\over 20},-{17\over 320}\>+u\<1,-{1\over 6},+{5\over 6}\>$}

\medskip\tabb The linear incidence conditions are

\medskip\tabb\tabb ${\bf G}_a^\pm,{\bf G}_b^\pm,{\bf G}_c^\pm,{\bf G}_{abc}^\pm,{\bf G}_{a+b}^{},{\bf H}_{b+c}^{},{\bf H}_{c-a}^{}$
\par\tabb\tabb $a_x^{} = u+{27\over 10}$

\medskip\tabb The lattice vectors are

\medskip\tabb\tabb $a = \<{27\over 10}+u,1+{1\over 6}u,-{1\over 10}+{11\over 6}u\>$
\par\tabb\tabb $b = \<-{3\over 10}-u,{13\over 5}-{1\over 6}u,{13\over 10}-{11\over 6}u\>$
\par\tabb\tabb $c = \<{9\over 10}-u,-{4\over 5}+{7\over 6}u,{23\over 10}+{5\over 6}u\>$
\par\tabb\tabb $d = \<{1\over 10}+u,0+{5\over 6}u,-{1\over 10}+{7\over 6}u\>$

\medskip\tabb The lattice volume and packing density are

\medskip\tabb\tabb $V = {9\over 1000}(4671+{25600\over 9}(u-{3\over 160})_{}^2)$
\par\tabb\tabb $\phi = 4000/(4671+{25600\over 9}(u-{3\over 160})_{}^2)$

\medskip{\bf C2$_{\rm max}^-$\tab Maximal line $\<0,-{1\over 20},+{3\over 64}\>+u\<1,-{1\over 6},+{1\over 2}\>$}

\medskip\tabb The linear incidence conditions are

\medskip\tabb\tabb ${\bf G}_a^\pm,{\bf G}_b^\pm,{\bf G}_c^\pm,{\bf G}_{abc}^\pm,{\bf G}_{a+b}^{},{\bf H}_{b-c}^{},{\bf H}_{c+a}^{}$
\par\tabb\tabb $a_x^{} = u+{27\over 10}$

\medskip\tabb The lattice vectors are

\medskip\tabb\tabb $a = \<{27\over 10}+u,{11\over 10}+{1\over 6}u,-{1\over 5}+{11\over 6}u\>$
\par\tabb\tabb $b = \<-{3\over 10}-u,{5\over 2}-{1\over 6}u,{7\over 5}-{11\over 6}u\>$
\par\tabb\tabb $c = \<{7\over 10}-{5\over 3}u,-{7\over 10}+{1\over 6}u,{12\over 5}+{1\over 2}u\>$
\par\tabb\tabb $d = \<{1\over 10}+u,-{1\over 10}+{5\over 6}u,0+{7\over 6}u\>$

\medskip\tabb The lattice volume and packing density are

\medskip\tabb\tabb $V = {9\over 1000}(4671+{25600\over 9}(u+{3\over 160})_{}^2)$
\par\tabb\tabb $\phi = 4000/(4671+{25600\over 9}(u+{3\over 160})_{}^2)$

\medskip{\bf C3$_{\rm orig}^{}$\tab Origin $\<0,0,0\>$}

\medskip\tabb The linear incidence conditions are

\medskip\tabb\tabb ${\bf G}_a^\pm,{\bf G}_b^\pm,{\bf G}_c^\pm,{\bf G}_{abc}^\pm,{\bf G}_{a+b}^{}$
\par\tabb\tabb $\<a_x^{},b_y^{},c_z^{}\> = \<u+{27\over 10},v+{51\over 20},w+{753\over 320}\>,\<u,v,w\> = \<0,0,0\>$

\medskip\tabb The lattice vectors are

\medskip\tabb\tabb $a = \<+{27\over 10},+{21\over 20},-{3\over 20}\>$
\par\tabb\tabb $b = \<-{3\over 10},+{51\over 20},+{27\over 20}\>$
\par\tabb\tabb $c = \<+{129\over 160},-{237\over 320},+{753\over 320}\>$
\par\tabb\tabb $d = \<+{1\over 10},-{1\over 20},-{1\over 20}\>$

\medskip\tabb The lattice volume and packing density are

\medskip\tabb\tabb $V = 2\det[a,b,c] = {1053\over 25}$
\par\tabb\tabb $\phi = 36/V = {100\over 117}$

\eject{\bf C3$_{\rm opt}^+$\tab Optimal point $\<+{3\over 160},+{3\over 64},-{3\over 80}\>$}

\medskip\tabb The linear incidence conditions are

\medskip\tabb\tabb ${\bf G}_a^\pm,{\bf G}_b^\pm,{\bf G}_c^\pm,{\bf G}_{abc}^\pm,{\bf G}_{a+b}^{},{\bf H}_{b+c}^{},{\bf H}_{c-a}^{},\partial V$
\par\tabb\tabb $a_x^{} = u+{27\over 10},u = +{3\over 160}$ [min $V$]

\medskip\tabb The lattice vectors are

\medskip\tabb\tabb $a = \<+{87\over 32},+{321\over 320},-{21\over 320}\>
= \<+2.71875,+1.003125,-0.065625\>$
\par\tabb\tabb $b = \<-{51\over 160},+{831\over 320},+{81\over 64}\>
= \<-0.31875,+2.596875,+1.265625\>$
\par\tabb\tabb $c = \<+{141\over 160},-{249\over 320},+{741\over 320}\>
= \<+0.88125,-0.778125,+2.315625\>$
\par\tabb\tabb $d = \<+{19\over 160},+{1\over 64},-{5\over 64}\>
= \<+0.11875,+0.015625,-0.078125\>$

\medskip\tabb The lattice volume and packing density are

\medskip\tabb\tabb $V = {42039\over 1000} = 42.039$
\par\tabb\tabb $\phi = {4000\over 4671}\approx .856347677156$

\medskip{\bf C3$_{\rm opt}^-$\tab Optimal point $\<-{3\over 160},-{3\over 64},+{3\over 80}\>$}

\medskip\tabb The linear incidence conditions are

\medskip\tabb\tabb ${\bf G}_a^\pm,{\bf G}_b^\pm,{\bf G}_c^\pm,{\bf G}_{abc}^\pm,{\bf G}_{a+b}^{},{\bf H}_{b-c}^{},{\bf H}_{c+a}^{},\partial V$
\par\tabb\tabb $a_x^{} = u+{27\over 10},u = -{3\over 160}$ [min $V$]

\medskip\tabb The lattice vectors are

\medskip\tabb\tabb $a = \<+{429\over 160},+{351\over 320},-{15\over 64}\>
= \<+2.68125,+1.096875,-0.234375\>$
\par\tabb\tabb $b = \<-{9\over 32},+{801\over 320},+{459\over 320}\>
= \<-0.28125,+2.503125,+1.434375\>$
\par\tabb\tabb $c = \<+{117\over 160},-{45\over 64},+{153\over 64}\>
= \<+0.73125,-0.703125,+2.390625\>$
\par\tabb\tabb $d = \<+{13\over 160},-{37\over 320},-{7\over 320}\>
= \<+0.08125,-0.115625,-0.021875\>$

\medskip\tabb The lattice volume and packing density are

\medskip\tabb\tabb $V = {42039\over 1000} = 42.039$
\par\tabb\tabb $\phi = {4000\over 4671}\approx .856347677156$

\medskip{\bf C3$_{\rm cen}^+$\tab Central point $\<0,+{1\over 20},-{17\over 320}\>$}

\medskip\tabb The linear incidence conditions are

\medskip\tabb\tabb ${\bf G}_a^\pm,{\bf G}_b^\pm,{\bf G}_c^\pm,{\bf G}_{abc}^\pm,{\bf G}_{a+b}^{},{\bf H}_{b+c}^{},{\bf H}_{c-a}^{},{\bf H}_{a-b}^{}$

\medskip\tabb The lattice vectors are

\medskip\tabb\tabb $a = \<+{27\over 10},+1,-{1\over 10}\>$
\par\tabb\tabb $b = \<-{3\over 10},+{13\over 5},+{13\over 10}\>$
\par\tabb\tabb $c = \<+{9\over 10},-{4\over 5},+{23\over 10}\>$
\par\tabb\tabb $d = \<+{1\over 10},0,-{1\over 10}\>$

\medskip\tabb The lattice volume and packing density are

\medskip\tabb\tabb $V = {5256\over 125}$
\par\tabb\tabb $\phi = {125\over 146}$

\eject{\bf C3$_{\rm cen}^-$\tab Central point $\<0,-{1\over 20},+{3\over 64}\>$}

\medskip\tabb The linear incidence conditions are

\medskip\tabb\tabb ${\bf G}_a^\pm,{\bf G}_b^\pm,{\bf G}_c^\pm,{\bf G}_{abc}^\pm,{\bf G}_{a+b}^{},{\bf H}_{b-c}^{},{\bf H}_{c+a}^{},{\bf H}_{a-b}^{}$

\medskip\tabb The lattice vectors are

\medskip\tabb\tabb $a = \<+{27\over 10},+{11\over 10},-{1\over 5}\>$
\par\tabb\tabb $b = \<-{3\over 10},+{5\over 2},+{7\over 5}\>$
\par\tabb\tabb $c = \<+{7\over 10},-{7\over 10},+{12\over 5}\>$
\par\tabb\tabb $d = \<+{1\over 10},-{1\over 10},0\>$

\medskip\tabb The lattice volume and packing density are

\medskip\tabb\tabb $V = {5256\over 125}$
\par\tabb\tabb $\phi = {125\over 146}$

\medskip{\bf C3$_{\rm sym}^+$\tab Symmetric point $\<0,0,+{3\over 64}\>$}

\medskip\tabb The linear incidence conditions are

\medskip\tabb\tabb ${\bf G}_a^\pm,{\bf G}_b^\pm,{\bf G}_c^\pm,{\bf G}_{abc}^\pm,{\bf G}_{a+b}^{},{\bf H}_{b-c}^{},{\bf H}_{c-a}^{},{\bf H}_{a-b}^{},\partial V$

\medskip\tabb The lattice vectors are

\medskip\tabb\tabb $a = \<+{27\over 10},+{21\over 20},-{3\over 20}\>$
\par\tabb\tabb $b = \<-{3\over 10},+{51\over 20},+{27\over 20}\>$
\par\tabb\tabb $c = \<+{9\over 10},-{3\over 5},+{12\over 5}\>$
\par\tabb\tabb $d = \<+{1\over 10},-{1\over 20},-{1\over 20}\>$

\medskip\tabb The lattice volume and packing density are

\medskip\tabb\tabb $V = {1053\over 25}$
\par\tabb\tabb $\phi = {100\over 117}$

\medskip{\bf C3$_{\rm sym}^-$\tab Symmetric point $\<0,0,-{33\over 320}\>$}

\medskip\tabb The linear incidence conditions are

\medskip\tabb\tabb ${\bf G}_a^\pm,{\bf G}_b^\pm,{\bf G}_c^\pm,{\bf G}_{abc}^\pm,{\bf G}_{a+b}^{},{\bf H}_{b+c}^{},{\bf H}_{c+a}^{},{\bf H}_{a-b}^{},\partial V$

\medskip\tabb The lattice vectors are

\medskip\tabb\tabb $a = \<+{27\over 10},+{21\over 20},-{3\over 20}\>$
\par\tabb\tabb $b = \<-{3\over 10},+{51\over 20},+{27\over 20}\>$
\par\tabb\tabb $c = \<+{3\over 5},-{21\over 20},+{9\over 4}\>$
\par\tabb\tabb $d = \<+{1\over 10},-{1\over 20},-{1\over 20}\>$

\medskip\tabb The lattice volume and packing density are

\medskip\tabb\tabb $V = {1053\over 25}$
\par\tabb\tabb $\phi = {100\over 117}$

\section{Small unit cell packings}

\medskip{\bf D1}\tab The clusters are

\medskip\tabb\tabb $+{\bf B}_1^{} = +B[o,p,q,r]$

\medskip\tabb The cluster vertices are

\medskip\tabb\tabb $o = \<+1,+1,+1\>$
\par\tabb\tabb $p = \<+1,-1,-1\>$
\par\tabb\tabb $q = \<-1,+1,-1\>$
\par\tabb\tabb $r = \<-1,-1,+1\>$

\medskip\tabb The neighbors of ${\bf B}_1^{}$ are

\medskip\tabb\tabb ${\bf B}_1^{}+a\tab{\bf B}_1^{}+b\tab{\bf B}_1^{}+c$
\par\tabb\tabb ${\bf B}_1^{}-a\tab{\bf B}_1^{}-b\tab{\bf B}_1^{}-c$
\par\tabb\tabb ${\bf B}_1^{}+a+b\tab{\bf B}_1^{}+b+c\tab{\bf B}_1^{}+c+a$
\par\tabb\tabb ${\bf B}_1^{}-a-b\tab{\bf B}_1^{}-b-c\tab{\bf B}_1^{}-c-a$
\par\tabb\tabb ${\bf B}_1^{}+a+b+c$
\par\tabb\tabb ${\bf B}_1^{}-a-b-c$

\medskip\tabb The linear incidence conditions are

\medskip\tabb\tabb $E[o,p]\cap(E[q,r]+a)\neq\emptyset$
\par\tabb\tabb $E[o,q]\cap(E[r,p]+b)\neq\emptyset$
\par\tabb\tabb $E[o,r]\cap(E[p,q]+c)\neq\emptyset$
\par\tabb\tabb $F[o,p,q]\cap(V[r]+a+b)\neq\emptyset$
\par\tabb\tabb $F[o,q,r]\cap(V[p]+b+c)\neq\emptyset$
\par\tabb\tabb $F[o,r,p]\cap(V[q]+c+a)\neq\emptyset$
\par\tabb\tabb $V[o]\cap(F[p,q,r]+a+b+c)\neq\emptyset$

\medskip\tabb\tabb\tabb optimize (3 free parameters)

\medskip\tabb The lattice vectors are

\medskip\tabb\tabb $a = \<+2,-{1\over 3},-{1\over 3}\>$
\par\tabb\tabb $b = \<-{1\over 3},+2,-{1\over 3}\>$
\par\tabb\tabb $c = \<-{1\over 3},-{1\over 3},+2\>$

\medskip\tabb The lattice volume and packing density are

\medskip\tabb\tabb $V = \det[a,b,c] = {196\over 27}$
\par\tabb\tabb $\phi = 1\cdot{8\over 3}/V = {18\over 49}$

\bigskip{\bf D2}\tab The clusters are

\medskip\tabb\tabb $+{\bf B}_1^{} = +B[o,p,q,r]$
\par\tabb\tabb $-{\bf B}_1^{} = -B[o,p,q,r]$

\medskip\tabb The cluster vertices are

\medskip\tabb\tabb $o = \<+1,+1,+1\>$
\par\tabb\tabb $p = \<+1,-1,-1\>$
\par\tabb\tabb $q = \<-1,+1,-1\>$
\par\tabb\tabb $r = \<-1,-1,+1\>$

\eject\tabb The neighbors of ${\bf B}_1^{}$ are

\medskip\tabb\tabb ${\bf B}_1^{}+a+b\tab{\bf B}_1^{}+b+c\tab{\bf B}_1^{}+c+a$
\par\tabb\tabb ${\bf B}_1^{}-a-b\tab{\bf B}_1^{}-b-c\tab{\bf B}_1^{}-c-a$
\par\tabb\tabb ${\bf B}_1^{}+2a+b+c\tab{\bf B}_1^{}+2b+c+a\tab{\bf B}_1^{}+2c+a+b$
\par\tabb\tabb ${\bf B}_1^{}-2a-b-c\tab{\bf B}_1^{}-2b-c-a\tab{\bf B}_1^{}-2c-a-b$
\par\tabb\tabb ${\bf B}_1^{}+2a+2b+2c$
\par\tabb\tabb ${\bf B}_1^{}-2a-2b-2c$

\medskip\tabb\tabb $d-{\bf B}_1^{}+a\tab d-{\bf B}_1^{}+b\tab d-{\bf B}_1^{}+c$
\par\tabb\tabb $d-{\bf B}_1^{}-a\tab d-{\bf B}_1^{}-b\tab d-{\bf B}_1^{}-c$
\par\tabb\tabb $d-{\bf B}_1^{}+a+b+c$
\par\tabb\tabb $d-{\bf B}_1^{}-a-b-c$

\medskip\tabb The linear incidence conditions are

\medskip\tabb\tabb $E[o,r]\cap(E[p,q]+a+b)\neq\emptyset$
\par\tabb\tabb $E[o,p]\cap(E[q,r]+b+c)\neq\emptyset$
\par\tabb\tabb $E[o,q]\cap(E[r,p]+c+a)\neq\emptyset$
\par\tabb\tabb $F[o,q,r]\cap(V[p]+2a+b+c)\neq\emptyset$
\par\tabb\tabb $E[o,r,p]\cap(V[q]+2b+c+a)\neq\emptyset$
\par\tabb\tabb $F[o,p,q]\cap(V[r]+2c+a+b)\neq\emptyset$
\par\tabb\tabb $V[o]\cap(F[p,q,r]+2a+2b+2c)\neq\emptyset$

\medskip\tabb\tabb $F[o,q,r]\cap(d-F[o,q,r]+a)\neq\emptyset$
\par\tabb\tabb $F[o,p,q]\cap(d-F[o,p,q]-b)\neq\emptyset$
\par\tabb\tabb $F[o,r,p]\cap(d-F[o,r,p]-c)\neq\emptyset$
\par\tabb\tabb $F[o,p,q]\cap(d-F[o,p,q]+a+b+c)\neq\emptyset$
\par\tabb\tabb $F[p,q,r]\cap(d-F[p,q,r]-a-b-c)\neq\emptyset$

\medskip\tabb\tabb\tabb optimize (2 free parameters)

\medskip\tabb The lattice vectors are

\medskip\tabb\tabb $a = {1\over 3}\<-7+\sqrt{10},-19+7\sqrt{10},+16-4\sqrt{10}\>$
\par\tabb\tabb $b = {1\over 3}\<+7-1\sqrt{10},-7+1\sqrt{10},-10+4\sqrt{10}\>$
\par\tabb\tabb $c = {1\over 3}\<-1+1\sqrt{10},+25-7\sqrt{10},+2-2\sqrt{10}\>$
\par\tabb\tabb $d = {1\over 3}\<-1+1\sqrt{10},+9-3\sqrt{10},-8+2\sqrt{10}\>$

\medskip\tabb The lattice volume and packing density are

\medskip\tabb\tabb $V = 2\det[a,b,c] = \det[a+b,b+c,c+a] = {16\over 27}(139-40\sqrt{10})$
\par\tabb\tabb $\phi = 2\cdot{8\over 3}/V = 9/(139-40\sqrt{10})$

\bigskip{\bf D3}\tab The clusters are

\medskip\tabb\tabb $+{\bf B}_1^{} = +B[o,p,q,r]$
\par\tabb\tabb $-{\bf B}_1^{} = -B[o,p,q,r]$
\par\tabb\tabb $T{\bf B}_1^{} = TB[o,p,q,r]$

\eject\tabb The cluster vertices are

\medskip\tabb\tabb $o = \<+1,+1,+1\>$
\par\tabb\tabb $p = \<+1,-1,-1\>$
\par\tabb\tabb $q = \<-1,+1,-1\>$
\par\tabb\tabb $r = \<-1,-1,+1\>$

\medskip\tabb\tabb\tabb $T = {1\over 3}\left[\scriptsize\matrixold{+1&-2&-2\cr -2&+1&-2\cr -2&-2&+1}\right]$\tabb reflection through the plane $+x+y+z = 0$

\medskip\tabb The neighbors of ${\bf B}_1^{}$ are

\medskip\tabb\tabb ${\bf B}_1^{}+a-b\tab{\bf B}_1^{}+b-c\tab{\bf B}_1^{}+c-a$
\par\tabb\tabb ${\bf B}_1^{}-a+b\tab{\bf B}_1^{}-b+c\tab{\bf B}_1^{}-c+a$
\par\tabb\tabb ${\bf B}_1^{}+2d$
\par\tabb\tabb ${\bf B}_1^{}-2d$

\medskip\tabb\tabb $-{\bf B}_1^{}-a\tab -{\bf B}_1^{}-a\tab -{\bf B}_1^{}-c$

\medskip\tabb\tabb $T{\bf B}_1^{}+d+a\tab T{\bf B}_1^{}+d+b\tab T{\bf B}_1^{}+d+c$
\par\tabb\tabb $T{\bf B}_1^{}-d+a\tab T{\bf B}_1^{}-d+b\tab T{\bf B}_1^{}-d+c$

\medskip\tabb The linear incidence conditions are

\medskip\tabb\tabb $E[r,p]\cap(V[q]+a-b)\neq\emptyset$
\par\tabb\tabb $E[p,q]\cap(V[r]+b-c)\neq\emptyset$
\par\tabb\tabb $E[q,r]\cap(V[p]+c-a)\neq\emptyset$
\par\tabb\tabb $V[o]\cap(F[p,q,r]+2d)\neq\emptyset$

\medskip\tabb\tabb $F[o,q,r]\cap(-F[o,q,r]-a)\neq\emptyset$
\par\tabb\tabb $F[o,r,p]\cap(-F[o,r,p]-b)\neq\emptyset$
\par\tabb\tabb $F[o,p,q]\cap(-F[o,p,q]-c)\neq\emptyset$

\medskip\tabb\tabb $E[o,p]\cap(TE[o,q]+d+a)\neq\emptyset$
\par\tabb\tabb $E[o,q]\cap(TE[o,r]+d+b)\neq\emptyset$
\par\tabb\tabb $E[o,r]\cap(TE[o,p]+d+c)\neq\emptyset$
\par\tabb\tabb $F[p,q,r]\cap(TF[p,q,r]-d+a)\neq\emptyset$
\par\tabb\tabb $F[p,q,r]\cap(TF[p,q,r]-d+b)\neq\emptyset$
\par\tabb\tabb $F[p,q,r]\cap(TF[p,q,r]-d+c)\neq\emptyset$

\medskip\tabb The lattice vectors are

\medskip\tabb\tabb $a = \<+1,-1,0\>$
\par\tabb\tabb $b = \<0,+1,-1\>$
\par\tabb\tabb $c = \<-1,0,+1\>$
\par\tabb\tabb $d = \<+{2\over 3},+{2\over 3},+{2\over 3}\>$

\medskip\tabb The lattice volume and packing density are

\medskip\tabb\tabb $V = 2\det[d+a,d+b,d+c] = 12$
\par\tabb\tabb $\phi = 3\cdot{8\over 3}/V = {2\over 3}$

\eject{\bf D5}\tab The clusters are

\medskip\tabb\tabb $+{\bf B}_5^{} = +({\bf F}_2^{}\cup{\bf B}_p^{}\cup{\bf B}_q^{}\cup{\bf B}_r^{})$

\medskip\tabb\tabb\tabb ${\bf F}_2^{} = B[o,p,q,r]\cup B[so,p,q,r]$
\par\tabb\tabb\tabb ${\bf B}_p^{} = d+T_p^{}B[sp,q,r,o]$
\par\tabb\tabb\tabb ${\bf B}_q^{} = e+T_q^{}B[sq,r,o,p]$
\par\tabb\tabb\tabb ${\bf B}_r^{} = f+T_r^{}B[sr,o,p,q]$

\medskip\tabb The cluster vertices are

\medskip\tabb\tabb $o = \<+1,+1,+1\>$
\par\tabb\tabb $p = \<+1,-1,-1\>$
\par\tabb\tabb $q = \<-1,+1,-1\>$
\par\tabb\tabb $r = \<-1,-1,+1\>$

\medskip\tabb\tabb\tabb $s = -{5\over 3}$

\medskip\tabb\tabb\tabb $T_{\<i,j,k\>}^\theta = (\cos\theta)\left[\scriptsize\matrixold{1&0&0\cr 0&1&0\cr 0&0&1}\right]+(\sin\theta)\left[\scriptsize\matrixold{0&-k&+\!j\cr +k&0&-i\cr -\!j&+i&0}\right]+(1-\cos\theta)\left[\scriptsize\matrixold{ii&ji&ki\cr i\!j&j\!j&k\!j\cr ik&jk&kk}\right]$
\par\tabb\tabb\tabb\tabb rotation by $\theta$ about the unit vector $\<i,j,k\>$

\medskip\tabb\tabb\tabb $T_p^{} = T_{\<-1,+1,+1\>/\sqrt 3}^\alpha$
\par\tabb\tabb\tabb $T_q^{} = T_{\<+1,-1,+1\>/\sqrt 3}^\beta$
\par\tabb\tabb\tabb $T_r^{} = T_{\<+1,+1,-1\>/\sqrt 3}^\gamma$

\medskip\tabb The neighbors of ${\bf B}_5^{}$ are

\medskip\tabb\tabb ${\bf B}_5^{}+a\tab{\bf B}_5^{}+b\tab{\bf B}_5^{}+c$
\par\tabb\tabb ${\bf B}_5^{}-a\tab{\bf B}_5^{}-b\tab{\bf B}_5^{}-c$
\par\tabb\tabb ${\bf B}_5^{}+a-b\tab{\bf B}_5^{}+b-c\tab{\bf B}_5^{}+c-a$
\par\tabb\tabb ${\bf B}_5^{}-a+b\tab{\bf B}_5^{}-b+c\tab{\bf B}_5^{}-c+a$

\medskip\tabb The linear incidence conditions are

\medskip\tabb\tabb $F[o,q,r]\cap(d+T_p^{}F[o,q,r])\neq\emptyset$
\par\tabb\tabb $F[o,r,p]\cap(e+T_q^{}F[o,r,p])\neq\emptyset$
\par\tabb\tabb $F[o,p,q]\cap(f+T_r^{}F[o,p,q])\neq\emptyset$

\medskip\tabb\tabb $(f+T_r^{}E[o,q])\cap(d+T_p^{}E[o,q])\neq\emptyset$
\par\tabb\tabb $(d+T_p^{}E[q,r])\cap(e+T_q^{}E[o,r])\neq\emptyset$
\par\tabb\tabb $(e+T_q^{}E[r,o])\cap(f+T_r^{}E[o,p])\neq\emptyset$

\medskip\tabb\tabb $E[o,p]\cap(E[so,q]+a)\neq\emptyset$
\par\tabb\tabb $(d+T_p^{}E[sp,q])\cap(E[so,r]+b)\neq\emptyset$
\par\tabb\tabb $(e+T_q^{}F[sq,o,p])\cap(V[so]+a)\neq\emptyset$
\par\tabb\tabb $(f+T_r^{}E[o,q])\cap(E[so,p]+b)\neq\emptyset$
\par\tabb\tabb $(d+T_p^{}E[sp,r])\cap(E[so,q]+c)\neq\emptyset$

\par\tabb\tabb $(e+T_q^{}E[sq,o])\cap(f+T_r^{}E[sr,p]+c)\neq\emptyset$
\par\tabb\tabb $(f+T_r^{}E[sr,o])\cap(e+T_q^{}E[sq,p]+b)\neq\emptyset$
\par\tabb\tabb $(d+T_p^{}E[r,o])\cap(f+T_r^{}E[sr,q]+c)\neq\emptyset$

\eject\tabb\tabb $(e+T_q^{}V[sq])\cap(d+T_p^{}F[sp,q,r]+a-b)\neq\emptyset$
\par\tabb\tabb $(f+T_r^{}E[sr,o])\cap(e+T_q^{}E[sq,p]+b-c)\neq\emptyset$
\par\tabb\tabb $(d+T_p^{}E[q,r])\cap(f+T_r^{}E[sr,o]+c-a)\neq\emptyset$
\par\tabb\tabb $F[so,q,r]\cap(f+T_r^{}V[sr]+c-a)\neq\emptyset$

\medskip\tabb\tabb\tabb optimize (3 free parameters)

\medskip\tabb The lattice vectors are

\medskip\tabb\tabb $a ≈ \<+2.59649957,+0.85417034,+1.06467162\>$
\par\tabb\tabb $b ≈ \<-0.14944338,+2.53196005,+0.58678219\>$
\par\tabb\tabb $c ≈ \<-0.02900844,-0.61508232,+2.49547175\>$
\par\tabb\tabb $d ≈ \<+0.02462012,-0.05626441,+0.08088453\>$
\par\tabb\tabb $e ≈ \<-0.03576085,-0.12004584,-0.08428499\>$
\par\tabb\tabb $f ≈ \<-0.04295985,+0.03620318,-0.00675668\>$

\medskip\tabb\tabb\tabb $\alpha ≈ +.07518716$
\par\tabb\tabb\tabb $\beta ≈ -.08591647$
\par\tabb\tabb\tabb $\gamma ≈ +.04543238$

\medskip\tabb The lattice volume and packing density are

\medskip\tabb\tabb $V = \det[a,b,c] ≈ 17.82301085$
\par\tabb\tabb $\phi = 5\cdot{8\over 3}/V ≈ .74809657$

\bigskip{\bf D6}\tab The clusters are

\medskip\tabb\tabb $+{\bf F}_2^{} = +(B[o,p,q,r]\cup B[p,q,r,s])$
\par\tabb\tabb $-{\bf F}_2^{} = -(B[o,p,q,r]\cup B[p,q,r,s])$
\par\tabb\tabb $-{\bf B}_1^{} = -B[o,p,q,r]$
\par\tabb\tabb $+{\bf B}_1^{} = +B[o,p,q,r]$

\medskip\tabb The cluster vertices are

\medskip\tabb\tabb $o = \<+1,+1,+1\>$
\par\tabb\tabb $p = \<+1,-1,-1\>$
\par\tabb\tabb $q = \<-1,+1,-1\>$
\par\tabb\tabb $r = \<-1,-1,+1\>$
\par\tabb\tabb $s = \<-{5\over 3},-{5\over 3},-{5\over 3}\>$

\medskip\tabb The neighbors of ${\bf F}_2^{}$ are

\medskip\tabb\tabb ${\bf F}_2^{}+a-b\tab{\bf F}_2^{}+b-c\tab{\bf F}_2^{}+c-a$
\par\tabb\tabb ${\bf F}_2^{}-a+b\tab{\bf F}_2^{}-b+c\tab{\bf F}_2^{}-c+a$
\par\tabb\tabb ${\bf F}_2^{}+d$
\par\tabb\tabb ${\bf F}_2^{}-d$

\medskip\tabb\tabb $e+f-{\bf F}_2^{}-a\tab e+f-{\bf F}_2^{}-b\tab e+f-{\bf F}_2^{}-c$

\medskip\tabb\tabb $e-{\bf B}_1^{}+a\tab e-{\bf B}_1^{}+b\tab e-{\bf B}_1^{}+c$

\medskip\tabb\tabb $f+{\bf B}_1^{}+a\tab f+{\bf B}_1^{}+b\tab f+{\bf B}_1^{}+c$

\eject\tabb The linear incidence conditions are

\medskip\tabb\tabb $V[p]\cap(F[o,q,r]+a-b)\neq\emptyset$
\par\tabb\tabb $V[q]\cap(F[o,r,p]+b-c)\neq\emptyset$
\par\tabb\tabb $V[r]\cap(F[s,p,q]+c-a)\neq\emptyset$
\par\tabb\tabb $E[o,p]\cap(E[s,r]+d)\neq\emptyset$

\medskip\tabb\tabb $F[s,q,r]\cap(e+f-F[s,q,r]-a)\neq\emptyset$
\par\tabb\tabb $F[s,r,p]\cap(e+f-F[s,r,p]-b)\neq\emptyset$
\par\tabb\tabb $F[s,p,q]\cap(e+f-F[s,p,q]-c)\neq\emptyset$

\medskip\tabb\tabb $F[o,p,q]\cap(e-F[o,p,q]+a)\neq\emptyset$
\par\tabb\tabb $F[o,q,r]\cap(e-F[o,q,r]+b)\neq\emptyset$
\par\tabb\tabb $F[o,r,p]\cap(e-F[o,r,p]+c)\neq\emptyset$

\medskip\tabb\tabb $E[o,p]\cap(f+E[q,r]+d+a)\neq\emptyset$
\par\tabb\tabb $E[o,q]\cap(f+E[r,p]+d+b)\neq\emptyset$
\par\tabb\tabb $E[o,r]\cap(f+E[p,q]+d+c)\neq\emptyset$

\medskip\tabb\tabb $(f+F[o,q,r]+d+a)\cap(e-F[o,q,r]+c)\neq\emptyset$

\medskip\tabb\tabb\tabb optimize (1 free parameter)

\medskip\tabb The lattice vectors are

\medskip\tabb\tabb $a = {1\over 4104}\<-3069+11\sqrt{396129},-2382+2\sqrt{396129},-3369-1\sqrt{396129}\>$
\par\tabb\tabb $b = {1\over 4104}\<+17694-34\sqrt{396129},+27651-37\sqrt{396129},+3777-7\sqrt{396129}\>$
\par\tabb\tabb $c = {1\over 4104}\<-14625+23\sqrt{396129},-25269+35\sqrt{396129},-408+8\sqrt{396129}\>$
\par\tabb\tabb $d = {1\over 4104}\<+22629-19\sqrt{396129},-30903+65\sqrt{396129},+4152+8\sqrt{396129}\>$
\par\tabb\tabb $e = {1\over 4104}\<+4131-5\sqrt{396129},+2382-2\sqrt{396129},-3777+7\sqrt{396129}\>$
\par\tabb\tabb $f = {1\over 1026}\<-2838+2\sqrt{396129},+2865-7\sqrt{396129},+1116-4\sqrt{396129}\>$

\medskip\tabb The lattice volume and packing density are

\medskip\tabb\tabb $V = \det[d+a,d+b,d+c] = {1\over 701784}(97802181-132043\sqrt{396129})$
\par\tabb\tabb $\phi = 6\cdot{8\over 3}/V = 11228544/(97802181-132043\sqrt{396129})$

\bigskip{\bf D10}\tab The clusters are

\medskip\tabb\tabb $+{\bf E}_5^{} = +({\bf B}_1^{}\cup{\bf E}_4^{})$
\par\tabb\tabb $-{\bf E}_5^{} = -({\bf B}_1^{}\cup{\bf E}_4^{})$

\medskip\tabb\tabb\tabb ${\bf B}_1^{} = B[t,u,v,w]$
\par\tabb\tabb\tabb ${\bf E}_4^{} = B[o,p,v,w]\cup B[p,q,v,w]\cup B[q,r,v,w]\cup B[r,s,v,w]$

\medskip\tabb The cluster vertices are

\medskip\tabb\tabb $o = \sqrt 6\;\<-{4\over 9},-{4\over 9},-{7\over 9}\>$
\par\tabb\tabb $p = \sqrt 6\;\<-{2\over 3},-{2\over 3},+{1\over 3}\>$
\par\tabb\tabb $q = \sqrt 6\;\<0,0,+1\>$
\par\tabb\tabb $r = \sqrt 6\;\<+{2\over 3},+{2\over 3},+{1\over 3}\>$
\par\tabb\tabb $s = \sqrt 6\;\<+{4\over 9},+{4\over 9},-{7\over 9}\>$

\eject\tabb\tabb $t = \<+1,+1,-2\>$
\par\tabb\tabb $u = \<-1,-1,-2\>$
\par\tabb\tabb $v = \<-1,+1,0\>$
\par\tabb\tabb $w = \<+1,-1,0\>$

\medskip\tabb The neighbors of ${\bf E}_5^{}$ are

\medskip\tabb\tabb ${\bf E}_5^{}+a+b\tab{\bf E}_5^{}+b+c\tab{\bf E}_5^{}+c+a$
\par\tabb\tabb ${\bf E}_5^{}-a-b\tab{\bf E}_5^{}-b-c\tab{\bf E}_5^{}-c-a$

\medskip\tabb\tabb $d-{\bf E}_5^{}+a\tab d-{\bf E}_5^{}+b\tab d-{\bf E}_5^{}+c$
\par\tabb\tabb $d-{\bf E}_5^{}-a\tab d-{\bf E}_5^{}-b\tab d-{\bf E}_5^{}-c$
\par\tabb\tabb $d-{\bf E}_5^{}+a+b+c$
\par\tabb\tabb $d-{\bf E}_5^{}-a-b-c$

\medskip\tabb The linear incidence conditions are

\medskip\tabb\tabb $V[t]\cap(F[p,q,v]+a+b)\neq\emptyset$
\par\tabb\tabb $E[r,v]\cap(E[p,w]+b+c)\neq\emptyset$
\par\tabb\tabb $F[q,r,w]\cap(V[u]+c+a)\neq\emptyset$

\medskip\tabb\tabb $F[t,u,v]\cap(d-F[t,u,v]+b)\neq\emptyset$
\par\tabb\tabb $F[t,u,w]\cap(d-F[t,u,w]-c)\neq\emptyset$
\par\tabb\tabb $F[o,p,v]\cap(d-F[o,p,v]-a)\neq\emptyset$
\par\tabb\tabb $F[r,s,w]\cap(d-F[r,s,w]+a)\neq\emptyset$
\par\tabb\tabb $F[q,r,v]\cap(d-F[q,r,v]+c)\neq\emptyset$
\par\tabb\tabb $F[p,q,w]\cap(d-F[p,q,w]-b)\neq\emptyset$
\par\tabb\tabb $F[r,s,v]\cap(d-F[r,s,v]+a+b+c)\neq\emptyset$
\par\tabb\tabb $F[o,p,w]\cap(d-F[o,p,w]-a-b-c)\neq\emptyset$

\medskip\tabb\tabb\tabb optimize (1 free parameter)

\medskip\tabb The lattice vectors are

\medskip\tabb\tabb $a = {1\over 21540}\<+16713+14822\sqrt 6,-25455+13840\sqrt 6,0\>$
\par\tabb\tabb $b = {1\over 7180}\<-2656+561\sqrt 6,+14440-2385\sqrt 6,-7285-4835\sqrt 6\>$
\par\tabb\tabb $c = {1\over 7180}\<-2656+561\sqrt 6,+14440-2385\sqrt 6,+7285+4835\sqrt 6\>$
\par\tabb\tabb $d = {1\over 7180}\<0,0,-4339+1889\sqrt 6\>$

\medskip\tabb The lattice volume and packing density are

\medskip\tabb\tabb $V = 2\det[a,b,c] = {1\over 370146232}(7885808912+1639809563\sqrt 6)$
\par\tabb\tabb $\phi = 10\cdot{8\over 3}/V = 29611698560/(23657426736+4919428689\sqrt 6)$

\end{appendix}

\end{document}